\documentclass[11pt]{article}
\usepackage{amssymb}

\usepackage{amsmath}

\newcounter{resultnum}[section]

\newcounter{conclusionnum}[section]

\newcounter{conditionnum}[section]

\newcounter{conjecturenum}[section]
\newtheorem{example}{Example}[section]

\newcounter{examplenum}[section]

\newcounter{exercisenum}[section]
\newtheorem{lemma}{Lemma}[section]

\newcounter{lemmanum}[section]

\newcounter{notationnum}[section]
\newtheorem{theorem}{Theorem}[section]

\newcounter{theoremnum}[section]
\newtheorem{definition}{Definition}[section]

\newcounter{definitionnum}[section]

\newcounter{corollarynum}[section]
\newtheorem{remark}{Remark}[section]

\newcounter{remarknum}[section]
\newtheorem{proposition}{Proposition}[section]

\newcounter{propositionnum}[section]
\newtheorem{convention}{Convention}[section]

\newcounter{conventionnum}[section]

\newcounter{acknowledgementnum}[section]

\newcounter{algorithmnum}[section]

\newcounter{axiomnum}[section]

\newcounter{casenum}[section]

\newcounter{claimnum}[section]

\newcounter{summarynum}[section]

\newcounter{problemnum}[section]
\newenvironment{proof}[1][]{\textbf{Proof.} }{}

\begin{document}

\title{ Spectral Functionals, Nonholonomic Dirac Operators, and
Noncommutative Ricci Flows}
\date{May 30, 2009}
\author{ Sergiu I. Vacaru\thanks{
Sergiu.Vacaru@gmail.com ;\ http://www.scribd.com/people/view/1455460-sergiu }
\\
%EndAName
{\quad} \\
{\small {\textsl{ "Al. I. Cuza"  Ia\c si University}, Science Department,}
}\\
{\small {\textsl{\ 54 Lascar Catargi  street, 700107, Ia\c si, Romania}} } }
\maketitle

\begin{abstract}
We formulate a noncommutative generalization of the Ricci flow theory in the framework of spectral action approach to noncommutative geometry. Grisha Perelman's functionals are generated as commutative versions of certain spectral functionals defined by nonholonomic Dirac operators and
corresponding spectral triples. We derive the formulas for spectral averaged energy and entropy functionals and state the conditions when such values describe (non)holonomic Riemannian configurations.

\vskip5pt

\textbf{Keywords:}\ Noncommutative geometry, nonholonomic manifolds,
spinors, Dirac operators, Ricci flows, Perelman functionals

\vskip3pt

MSC2000:\ 53C44, 58B34, 58B20, 83C65

PACS2008:\ 02.40.Gh, 04.90.+e, 11.10.Nx, 11.30.Ly
\end{abstract}

%\newpage

\tableofcontents

\section{ Introduction}

The Ricci flow equations \cite{ham1} and Perelman functionals \cite{gper1}
can be re--defined with respect to moving frames subjected to nonholonomic
constraints \cite{vnhrf2}.\footnote{%
there are used also some other equivalent terms like anholonomic, or
non--integrable, restrictions/ constraints; we emphasize that in classical
and quantum physics the field and evolution equations play a fundamental
role but together with certain types of constraints and broken symmetries; a
rigorous mathematical approach to modern physical theories can be elaborated
only following geometric methods from 'nonholonomic field theory and
mechanics'} Considering models of evolution of geometric objects in a form
adapted to certain classes of nonholonomic constraints, we proved that
metrics and connections defining (pseudo) Riemannian spaces may flow into
similar nonholonomically deformed values modelling generalized Finsler and
Lagrange configurations \cite{vrfg}, with symmetric and nonsymmetric
metrics, or possessing noncommutative symmetries \cite{vncg}.

The original Hamilton--Perelman constructions were for unconstrained flows
of metrics evolving only on (pseudo) Riemannian manifolds. There were proved
a set of fundamental results in mathematics and physics (for instance,
Thurston and Poincar\'{e} conjectures, related to spacetime topological
properties, Ricci flow running of physical constants and fields etc), see
Refs. \cite{caozhu,cao,kleiner,rbook}, for reviews of mathematical results,
and \cite{vrf,vv1,vnhrf2}, for some applications in physics. Nevertheless, a
number of important problems in geometry and physics are considered in the
framework of classical and quantum field theories with constraints (for
instance, the Lagrange and Hamilton mechanics, Dirac quantization of
constrained systems, gauge theories with broken symmetries etc). With
respect to the Ricci flow theory, to impose constraints on evolution
equations is to extend the research programs on manifolds enabled with
nonholonomic distributions, i.e. to study flows of fundamental geometric
structures on nonholonomic manifolds.\footnote{%
on applications of the geometry of nonholonomic manifolds (which are
manifolds enabled with nonholonomic distributions/structures) to standard
theories of physics, see Refs. \cite{vsgg,vrfg};  historical reviews of
results and applications to Finsler, Lagrange, Hamilton geometry and
generalizations are considered in Refs. \cite{bejf,ma,mhohs}}

Imposing certain noncommutative conditions on physical variables and
coordinates in an evolution theory, we transfer the constructions and
methods into the field of noncommutative geometric analysis on nonholonomic
manifolds. This also leads naturally to various problems related to
noncommutative generalizations of the Ricci flow theory and possible
applications in modern physics.

In this work, \ we follow the approach to noncommutative geometry when the
spectral action paradigm \cite{connes1}, with spectral triples and Dirac
operators, gives us a very elegant formulation of the standard model in
physics. We cite here some series of works on noncommutative Connes--Lott
approach to the standard model and further developments and alternative
approaches \cite{conlot,chcon1,vancea,seiberg,
castro1,moffat,jurco,nishino,caciatori,aschieri,szabo,zerjak,chcon2}, see
also a recent review of results in Ref. \cite{chamconmarc} and monographs \cite{landi,madore,fgbv}.

Following the spectral action paradigm, all details of the standard models
of particle interactions and gravity can be ''extracted'' from a
noncommutative geometry generated by a spectral triple $(\mathcal{A},%
\mathcal{H},\mathcal{D})$ by postulating the action%
\begin{equation}
Tr~~f(\mathcal{D}^{2}/\Lambda ^{2})+<\Psi |\mathcal{D}|\Psi >,  \label{tract}
\end{equation}%
where ''spectral'' is in the sense that the action depends only on the
spectrum of the Dirac operator $\mathcal{D}$ on a certain noncommutative
space defined by a noncommutative associative algebra $\mathcal{A}=C^{\infty
}(V)\otimes ~^{P}\mathcal{A}.$ \footnote{%
we use a different system of notations than that in \cite{chcon2,chamconmarc}
because we have to adapt the formalism to denotations and methods formally
elaborated in noncommutative geometry, Ricci flow theory and nonholonomic
geometries} In formula (\ref{tract}), $Tr$ is the trace in operator algebra
and $\Psi $ is a spinor, all defined for a Hilbert space $\mathcal{H}$, $%
\Lambda $ is a cutoff scale and $~f$ \ is a positive function. For a number
of physical applications, $~^{P}\mathcal{A}$ is a finite dimensional algebra
and $C^{\infty }(V)$ is the algebra of complex valued and smooth functions
over a ''space'' $V,$ a topological manifold, which for different purposes
can be enabled with various necessary geometric structures. The spectral
geometry of $\mathcal{A}$ is given by the product rule $\mathcal{H}%
=L^{2}(V,S)\otimes ~^{P}\mathcal{H},$ where $L^{2}(V,S)$ is the Hilbert
space of $L^{2}$ spinors and $~^{P}\mathcal{H}$ is the Hilbert space of
quarks and leptons fixing the choice of the Dirac operator $~^{P}D$ and the
action $~^{P}\mathcal{A}$ for fundamental particles. Usually, the Dirac
operator from (\ref{tract}) is parametrized $\mathcal{D}=~^{V}D\otimes
1+\gamma _{5}\otimes ~^{P}D,$ where $~^{V}D$ is the Dirac operator of the
Levi--Civita spin connection on $V.$\footnote{%
in this work, we shall use left ''up'' and ''low'' abstract labels which
should not be considered as tensor or spinor indices written in the right
side of symbols for geometrical objects}

In order to construct exact solutions with noncommutative symmetries and
noncommutative gauge models of gravity \cite{vncg,vncgr} and include dilaton
fields \cite{chcon3}, one has to use instead of $~^{V}D$ certain generalized
types of Dirac operators defined by nonholonomic and/or conformal
deformations of the 'primary' Levi--Civita spin connection. In a more
general context, the problem of constructing well defined geometrically and
physically motivated nonholonomic Dirac operators is related to the problem
of definition of spinors and Dirac operators on Finsler--Lagrange spaces and
generalizations \cite{vspinor,vhspinor,vcliffordalg}; for a review of
results see \cite{vstav,vvicol} and Part III in the collection of works \cite%
{vsgg}, containing a series of papers and references on noncommutative
generalizations of Riemann--Finsler and Lagrange--Hamilton geometries,
nonholonomic Clifford structures and Dirac operators and applications to
standard models of physics and string theory.

The aims and results of this article are outlined as follow: \ Section \ref%
{s2} is devoted to an introduction to the geometry of nonholonomic
(commutative) Riemannian manifolds and definition of spinors on such
manifolds. Nonholonomic Dirac operators and related spectral triples are
considered in Section \ref{s3}. We show how to compute distances in such
nonholonomic spinor spaces. The main purpose of this paper (see Section \ref%
{s4}) is to prove that the Perelman's functionals \cite{gper1} and their
generalizations for nonholonomic Ricci flows in \cite{vnhrf2} can be
extracted from corresponding spectral functionals defining flows of a
generalized Dirac operator and their scalings.\ Finally, in Section \ref{s5}
we discuss and conclude the results of the paper. Certain important
component formulas are outlined in Appendix.

\section{Nonholonomic Manifolds and Spinor Structures}

\label{s2}The concept of nonholonomic manifold was introduced independently
by G. Vr\v{a}nceanu \cite{vr1} and Z. Horak \cite{hor} for geometric
interpretations of nonholonomic mechanical systems (see modern approaches
and historical remarks in Refs. \cite{bejf,vsgg,grozleit}). They called a
pair $(\mathbf{V},\mathcal{N})$, where $\mathbf{V}$ is a manifold and $%
\mathcal{N}$ is a nonintegrable distribution on $\mathbf{V}$, to be a
nonholonomic manifold and considered new classes of linear connections
(which were different from the Levi--Civita connection). Three well known
classes of nonholonomic manifolds, when the nonholonomic distribution
defines a nonlinear connection (N--connection) structure, are defined by
Finsler spaces \cite{cartan,bejancu,bcs} and their generalizations as
Lagrange and Hamilton spaces and higher order models \cite{kern,ma,mhohs,bm}
(usually such geometries are modelled on a tangent bundle). More recent
examples, related to exact off--diagonal solutions and nonholonomic frames
in Einstein/ string/ gauge/ quantum/ noncommutative gravity and nonholonomic
Fedosov manifolds \cite{vncg,esv,vsgg,vpla} also emphasize nonholonomic
geometric structures but on generic spacetime manifolds and generalizations.

The aim of this section is to formulate the geometry of nonholonomic
Clifford structures in a form adapted to generalizations for noncommutative
spaces.

\subsection{Nonholonomic distributions and nonlinear connections}

We consider a $(n+m)$--dimensional manifold $\mathbf{V,}$ with $n\geq 2$ and
$m\geq 1$ (for simplicity, in this work, $\mathbf{V}$ is a real \ smooth
Riemannian space). The local coordinates on $\mathbf{V}$ are denoted $%
u=(x,y),$ or $u^{\alpha }=\left( x^{i},y^{a}\right) ,$ where the
''horizontal'' (h) indices run the values $i,j,k,\ldots =1,2,\ldots ,n$ and
the ''vertical'' (v) indices run the values $a,b,c,\ldots =n+1,n+2,\ldots
,n+m.$ We parameterize a metric structure on $\mathbf{V}$ in the form

\begin{equation}
\mathbf{\ g}=\underline{g}_{\alpha \beta }\left( u\right) du^{\alpha
}\otimes du^{\beta }  \label{metr}
\end{equation}%
defined with respect to a local coordinate basis $du^{\alpha }=\left(
dx^{i},dy^{a}\right) $ by coefficients%
\begin{equation}
\underline{g}_{\alpha \beta }=\left[
\begin{array}{cc}
g_{ij}\left( u\right) +\underline{N}_{i}^{a}\left( u\right) \underline{N}%
_{j}^{b}\left( u\right) h_{ab}\left( u\right) & \underline{N}_{j}^{e}\left(
u\right) h_{ae}\left( u\right) \\
\underline{N}_{i}^{e}\left( u\right) h_{be}\left( u\right) & h_{ab}\left(
u\right)%
\end{array}%
\right] .  \label{ansatz}
\end{equation}

We denote by $\pi ^{\top }:T\mathbf{V}\rightarrow TV$ the differential of a
map $\mathbf{V}\rightarrow V$ \ defined by fiber preserving morphisms of the
tangent bundles $T\mathbf{V}$ and $TV,$ where $V$ is a $n$--dimensional
manifold of necessary smooth class.\footnote{%
For simplicity, we restrict our considerations for a subclass of
nonholonomic distributions $\mathcal{N}$ modelling certain fibered structures $%
\mathbf{V}\rightarrow V$ \ with constant rank $\pi .$ In such a case, \ the
map $\pi ^{\top }$ is similar to that for a vector bundle with total space $%
\mathbf{V}$ and base $V.$ In general, we can use any map $\pi ^{\top }$ for
which the kernel defines a corresponding vertical subspace as a nonholonomic
distribution.} The kernel of $\pi ^{\top }$ is just the vertical subspace $v%
\mathbf{V}$ with a related inclusion mapping $i:v\mathbf{V}\rightarrow T%
\mathbf{V}.$

\begin{definition}
A nonlinear connection (N--connection) $\mathbf{N}$ on a manifold $\mathbf{V}
$ is defined by the splitting on the left of an exact sequence
\begin{equation*}
0\rightarrow v\mathbf{V}\overset{i}{\rightarrow} T\mathbf{V}\rightarrow T%
\mathbf{V}/v\mathbf{V}\rightarrow 0,
\end{equation*}%
i. e. by a morphism of submanifolds $\mathbf{N:\ \ }T\mathbf{V}\rightarrow v%
\mathbf{V}$ such that $\mathbf{N\circ i}$ is the unity in $v\mathbf{V}.$
\end{definition}

Locally, a N--connection is defined by its coefficients $N_{i}^{a}(u),$%
\begin{equation}
\mathbf{N}=N_{i}^{a}(u)dx^{i}\otimes \frac{\partial }{\partial y^{a}}.
\label{coeffnc}
\end{equation}%
Globalizing the local splitting, one prove that any N--connection is defined
by a Whitney sum of conventional horizontal (h) subspace, $\left( h\mathbf{V}%
\right) ,$ and vertical (v) subspace, $\left( v\mathbf{V}\right) ,$
\begin{equation}
T\mathbf{V}=h\mathbf{V}\oplus v\mathbf{V}.  \label{whitney}
\end{equation}

The sum (\ref{whitney}) states on $T\mathbf{V}$ a nonholonomic distribution
of horizontal and vertical subspaces. The well known class of linear
connections consists on a particular subclass with the coefficients being
linear on $y^{a},$ i.e. $N_{i}^{a}(u)=\Gamma _{bj}^{a}(x)y^{b}.$

For simplicity, we shall work with a particular class of nonholonomic
manifolds:

\begin{definition}
\label{defanhm} A manifold $\mathbf{V}$ is N--anholonomic if its tangent
space $T\mathbf{V}$ is enabled with a N--connection structure (\ref{whitney}%
).
\end{definition}

There are also two important examples of N--anholonomic manifolds modelled on bundle spaces:

\begin{example}
\label{enhvtb} a) A vector bundle $\mathbf{E}=(E,\pi ,M,\mathbf{N})$ defined
by a surjective projection $\pi :E\rightarrow M,$ with $M$ being the base
manifold, $\dim M=n,$ and $E$ being the total space, $\dim E=n+m,$ and
provided with a N--connection splitting (\ref{whitney}) is a N--anholonomic
vector bundle.

b) A particular case is that of N--anholonomic tangent bundle $\mathbf{TM}%
=(TM,\pi ,M,\mathbf{N}),$ with dimensions $n=m.$\footnote{%
For the tangent bundle $TM,$ we can consider that both type of indices run
the same values.}

Following our unified geometric formalism, we can write that in the above
mentioned examples $\mathbf{V=E,}$ or $\mathbf{V=TM.}$
\end{example}

A N--anholonomic manifold is characterized by its curvature:

\begin{definition}
The N--connection curvature is defined as the Neijenhuis tensor, %
%\begin{equation*}
$\mathbf{\Omega }(\mathbf{X,Y})\doteqdot \lbrack v\mathbf{X},v\mathbf{Y}]+\ v[%
\mathbf{X,Y}]- v[v\mathbf{X},\mathbf{Y}]-v[\mathbf{X},v\mathbf{Y}].$ 
%\label{njht}
%\end{equation*}
\end{definition}

In local form, we have for $\mathbf{\Omega }=\frac{1}{2}\Omega _{ij}^{a}\
d^{i}\wedge d^{j}\otimes \partial _{a}$ the coefficients%
\begin{equation}
\Omega _{ij}^{a}=\frac{\partial N_{i}^{a}}{\partial x^{j}}-\frac{\partial
N_{j}^{a}}{\partial x^{i}}+N_{i}^{b}\frac{\partial N_{j}^{a}}{\partial y^{b}}%
-N_{j}^{b}\frac{\partial N_{i}^{a}}{\partial y^{b}}.  \label{ncurv}
\end{equation}

Performing a frame (vielbein) transform $\mathbf{e}_{\alpha }=\mathbf{e}%
_{\alpha }^{\ \underline{\alpha }}\partial _{\underline{\alpha }}$ and $%
\mathbf{e}_{\ }^{\beta }=\mathbf{e}_{\ \underline{\beta }}^{\beta }du^{%
\underline{\beta }},$ where we underline the local coordinate indices, when $%
\partial _{\underline{\alpha }}=\partial /\partial u^{\underline{\alpha }%
}=(\partial _{\underline{i}}=\partial /\partial x^{\underline{i}},\partial
/\partial y^{\underline{a}}),$ with coefficients
\begin{equation*}
\mathbf{e}_{\alpha }^{\ \underline{\alpha }}(u)=\left[
\begin{array}{cc}
e_{i}^{\ \underline{i}}(u) & N_{i}^{b}(u)e_{b}^{\ \underline{a}}(u) \\
0 & e_{a}^{\ \underline{a}}(u)%
\end{array}%
\right] ,~\mathbf{e}_{\ \underline{\beta }}^{\beta }(u)=\left[
\begin{array}{cc}
e_{\ \underline{i}}^{i\ }(u) & -N_{k}^{b}(u)e_{\ \underline{i}}^{k\ }(u) \\
0 & e_{\ \underline{a}}^{a\ }(u)%
\end{array}%
\right] ,
\end{equation*}%
we transform the metric (\ref{whitney}) into a distinguished metric
(d--metric)
\begin{equation}
\mathbf{g}=~^{h}g+~^{v}h=\ g_{ij}(x,y)\ e^{i}\otimes e^{j}+\ h_{ab}(x,y)\
\mathbf{e}^{a}\otimes \mathbf{e}^{b},  \label{m1}
\end{equation}%
for an associated, to a N--connection, frame (vielbein) structure $\mathbf{e}%
_{\nu }=(\mathbf{e}_{i},e_{a}),$ where
\begin{equation}
\mathbf{e}_{i}=\frac{\partial }{\partial x^{i}}-N_{i}^{a}(u)\frac{\partial }{%
\partial y^{a}}\mbox{ and
}e_{a}=\frac{\partial }{\partial y^{a}},  \label{dder}
\end{equation}%
and the dual frame (coframe) structure $\mathbf{e}^{\mu }=(e^{i},\mathbf{e}%
^{a}),$ where
\begin{equation}
e^{i}=dx^{i}\mbox{ and }\mathbf{e}^{a}=dy^{a}+N_{i}^{a}(u)dx^{i}.
\label{ddif}
\end{equation}

A vector field $\mathbf{X}\in T\mathbf{V}$ \ can be expressed
\begin{equation*}
\mathbf{X}=(hX,\ vX),\mbox{ \ or \ }\mathbf{X}=X^{\alpha }\mathbf{e}_{\alpha
}=X^{i}\mathbf{e}_{i}+X^{a}e_{a},
\end{equation*}%
where $hX=X^{i}\mathbf{e}_{i}$ and $vX=X^{a}e_{a}$ state, respectively, the
adapted to the N--connection structure horizontal (h) and vertical (v)
components of the vector. In brief, $\mathbf{X}$ is called a distinguished
vector (in brief, d--vector). \footnote{%
The vielbeins (\ref{dder}) and (\ref{ddif}) are called respectively
N--adapted frames and coframes. In order to preserve a relation with some
previous our notations \cite{vncg,vsgg}, we emphasize that $\mathbf{e}_{\nu
}=(\mathbf{e}_{i},e_{a})$ and $\mathbf{e}^{\mu }=(e^{i},\mathbf{e}^{a})$ are
correspondingly the former ''N--elongated'' partial derivatives $\delta
_{\nu }=\delta /\partial u^{\nu }=(\delta _{i},\partial _{a})$ and
''N--elongated'' differentials $\delta ^{\mu }=\delta u^{\mu }=(d^{i},\delta
^{a}).$ They define certain ``N--elongated'' differential operators which
are more convenient for tensor and integral calculations on such
nonholonomic manifolds.}

\begin{convention}
The geometric objects on $\mathbf{V}$ like tensors, spinors, connections etc
are called respectively d--tensors, d--spinors, d--connections etc if they
are adapted to the N--connection splitting (\ref{whitney}).
\end{convention}

The vielbeins (\ref{ddif}) satisfy the nonholonomy relations
\begin{equation}
\lbrack \mathbf{e}_{\alpha },\mathbf{e}_{\beta }]=\mathbf{e}_{\alpha }%
\mathbf{e}_{\beta }-\mathbf{e}_{\beta }\mathbf{e}_{\alpha }=W_{\alpha \beta
}^{\gamma }\mathbf{e}_{\gamma }  \label{anhrel}
\end{equation}%
with (antisymmetric) nontrivial anholonomy coefficients $W_{ia}^{b}=\partial
_{a}N_{i}^{b}$ and $W_{ji}^{a}=\Omega _{ij}^{a}.$

On any commutative nonholonomic manifold $\mathbf{V,}$ we can work
equivalently with an infinite number of d--connections $~^{N}\mathbf{D=}~^{N}%
\mathbf{\Gamma (g),}$ which are d--metric compatible, $^{N}\mathbf{D~g=0,}$
and uniquely defined by a given metric $\mathbf{g.}$ Writing the deformation
relation
\begin{equation}
~^{N}\mathbf{\Gamma (g)=}~\ \Gamma (\mathbf{g})+~^{N}\mathbf{Z}(\mathbf{g}),
\label{ddcon}
\end{equation}%
where the deformation tensor $~^{N}\mathbf{Z}(\mathbf{g})$ is also uniquely
defined by $\mathbf{g,}$ we can transform any geometric construction for the
Levi--Civita connection $~\Gamma (\mathbf{g})$ equivalently into
corresponding constructions for the d--connection $~^{N}\mathbf{\Gamma (g),}$
and inversely, see details in Refs. \cite{vnhrf2,vrfg,vncg,vsgg,vijmmp}.
From a formal point of view, there is a nontrivial torsion $~^{N}\mathbf{T}(%
\mathbf{g}),$ computed following formula (\ref{tors1}), with coefficients (%
\ref{dtors}), all considered for $~^{N}\mathbf{\Gamma (g).}$ This torsion is
induced nonholonomically as an effective one (by anholonomy coefficients,
see (\ref{anhrel}) and (\ref{ncurv})) and constructed only from the
coefficients of metric $\mathbf{g}.$ Such a torsion is completely deferent
from that in string, or Einstein--Cartan, theory when the torsion tensor is
an additional (to metric) field defined by an antisymmetric $H$--field, or
spinning matter, see Ref. \cite{string1}.

The main conclusion of this section is that working with nonholonomic
distributions on formal Riemannian manifolds we can model, by anholonomic
frames and adapted geometric objects, various types of geometric structures
and physical theories with generic off--diagonal gravitational interactions,
constrained Lagrange--Hamilton dynamics, Finsler and Lagrange spaces etc.

\subsection{N--anholonomic spin structures}

The spinor bundle on a manifold $M,\ dimM=n,$ is constructed on the tangent
bundle $TM$ by substituting the group $SO(n)$ by its universal covering $%
Spin(n).$ If a horizontal quadratic form $\ ^{h}g_{ij}(x,y)$ is defined on $%
T_{x}h\mathbf{V}$ we can consider h--spinor spaces in every point $x\in h%
\mathbf{V}$ with fixed $y^{a}.$ The constructions can be completed on $T%
\mathbf{V}$ by using the d--metric $\ \mathbf{g}$ (\ref{m1}). In this case,
the group $SO(n+m)$ is not only substituted by $Spin(n+m)$ but with respect
to N--adapted frames (\ref{dder}) and (\ref{ddif}) there are emphasized
decompositions to $Spin(n)\oplus Spin(m).$\footnote{%
It should be noted here that spin bundles may not exist for general
holonomic or nonholonomic manifolds. For simplicity, we do not provide such
topological considerations in this paper, see Ref. \cite{vgonz} on nontrivial topological configurations with nonholonmic manifolds. We state that we shall work only
with N--anholonomic manifolds for which certain spinor structures can be
defined both for the h- and v--splitting; the existence of a well defined
decomposition $Spin(n)\oplus Spin(m)$ follows from N--connection splitting (%
\ref{whitney}).}

\subsubsection{Clifford N--adapted modules (d--modules)}

A Clifford d--algebra is a\ $\ \wedge V^{n+m}$ algebra endowed with a
product
\begin{equation*}
\mathbf{u}\mathbf{v}+\mathbf{v}\mathbf{u}=2\mathbf{g}(\mathbf{u},\mathbf{v}%
)\ \mathbb{I}
\end{equation*}%
distinguished into h--, v--products
\begin{eqnarray*}
~^{h}u~^{h}v+~^{h}v~^{h}u &=&2~^{h}g(u,v)\ ~^{h}\mathbb{I}\ , \\
~^{v}u\ \ ~^{v}v+~^{v}v\ \ ~^{v}u &=&2\ \ ~^{v}h(~^{v}u,\ \ ~^{v}v)~^{v}%
\mathbb{I}\ ,\
\end{eqnarray*}%
for any $\mathbf{u}=(~^{h}u,\ ~^{v}u),\ \mathbf{v}=(~^{h}v,\ ~^{v}v)\in
V^{n+m},$ where $\mathbb{I},$ $\ ~^{h}\mathbb{I}\ $\ and $^{v}\mathbb{I}\ \ $%
\ are unity matrices of corresponding dimensions $(n+m)\times (n+m),$ or $%
n\times n$ and $m\times m.$

A metric $^{h}g$ on $h\mathbf{V}$ is defined by sections of the tangent
space $T~h\mathbf{V}$ provided with a bilinear symmetric form on continuous
sections $\Gamma (T~h\mathbf{V}).$\footnote{%
for simplicity, we shall consider only ''horizontal'' geometric
constructions if they are similar to ''vertical'' ones} This allows us to
define Clifford h--algebras $~^{h}\mathcal{C}l(T_{x}h\mathbf{V}),$ in any
point $x\in T~h\mathbf{V},$
\begin{equation*}
\gamma _{i}\gamma _{j}+\gamma _{j}\gamma _{i}=2\ g_{ij}~^{h}\mathbb{I}\ .
\end{equation*}%
For any point $x\in h\mathbf{V}$ and fixed $y=y_{0},$ there exists a standard
complexification, $T_{x}h\mathbf{V}^{\mathbb{C}}\doteq T_{x}h\mathbf{V}%
+iT_{x}h\mathbf{V},$ which can be used for definition of the 'involution'
operator on sections of $T_{x}h\mathbf{V}^{\mathbb{C}},$
\begin{equation*}
~^{h}\sigma _{1}~^{h}\sigma _{2}(x)\doteq ~^{h}\sigma _{2}(x)~^{h}\sigma
_{1}(x),\ ~^{h}\sigma ^{\ast }(x)\doteq ~^{h}\sigma (x)^{\ast },\forall x\in
h\mathbf{V},
\end{equation*}%
where ''*'' denotes the involution on every $~^{h}\mathcal{C}l(T_{x}h\mathbf{%
V}).$

\begin{definition}
A Clifford d--space on a nonholonomic manifold $~\mathbf{V}$ enabled with a
d--metric $\ \mathbf{g}(x,y)$ (\ref{m1}) and a N--connection $\ \mathbf{N}$ (%
\ref{coeffnc}) is defined as a Clifford bundle $~\mathcal{C}l(\mathbf{V}%
)=~^{h}\mathcal{C}l(h\mathbf{V})\oplus ~^{v}\mathcal{C}l(v\mathbf{V}),$ for
the $\ $Clifford h--space $~^{h}\mathcal{C}l(h\mathbf{V})\doteq ~^{h}%
\mathcal{C}l(T^{\ast }h\mathbf{V})$ and Clifford v--space $^{v}\mathcal{C}l(v%
\mathbf{V})\doteq ~^{v}\mathcal{C}l(T^{\ast }v\mathbf{V}).$
\end{definition}

For a fixed N--connection structure, a Clifford N--anholonomic bundle on $%
\mathbf{V}$ is defined $\ ~^{N}\mathcal{C}l(\mathbf{V})\doteq ~^{N}\mathcal{C%
}l(T^{\ast }\mathbf{V}).$ Let us consider a complex vector bundle $~^{E}\pi :\ E\rightarrow \mathbf{V}$
on an N--anholonomic space $\mathbf{V}$ when the N--connection structure is
given for the base manifold. The Clifford d--module of a vector bundle ${E}$
is defined by the $C(\mathbf{V})$--module $\Gamma ({E})$ of continuous
sections in ${E},$ 
%\begin{equation*}
$c:\ \Gamma (~^{N}\mathcal{C}l(\mathbf{V}))\rightarrow End(\Gamma ({E})).$
%\end{equation*}

In general, a vector bundle on a N--anholonomic manifold may be not adapted
to the N--connection structure on base space.

\subsubsection{h--spinors,  v--spinors and d--spinors}

Let us consider a vector space $V^{n}$ provided with Clifford structure. We
denote such a space $^{h}V^{n}$ in order to emphasize that its tangent space
is provided with a quadratic form $\ ^{h}g.$ We also write $~^{h}\mathcal{C}%
l(V^{n})\equiv \mathcal{C}l(~^{h}V^{n})$ and use subgroup $SO(\
^{h}V^{n})\subset O(\ ^{h}V^{n}).$

\begin{definition}
The space of complex h--spins is defined by the subgroup $\ $%
\begin{equation*}
^{h}Spin^{c}(n)\equiv Spin^{c}(\ ^{h}V^{n})\equiv \
^{h}Spin^{c}(V^{n})\subset \mathcal{C}l(\ ^{h}V^{n}),
\end{equation*}
determined by the products of pairs of vectors $w\in \ ^{h}V^{\mathbb{C}}$
when $w\doteq \lambda u$ where $\lambda $ is a complex number of module 1
and $u$ is of unity length in $\ ^{h}V^{n}.$
\end{definition}

Similar constructions can be performed for the v--subspace $~^{v}V^{m},$
which allows us to define similarly the group of real v--spins.

A usual spinor is a section of a vector bundle $S$ on a manifold $M$ when an
irreducible representation of the group $Spin(M)\doteq Spin(T_{x}^{\ast }M)$
is defined on the typical fiber. The set of sections $\Gamma (S)$ is a
irreducible Clifford module. If the base manifold is of type $h\mathbf{V},$
or is a general N--anholonomic manifold $\mathbf{V},$ we have to define
spinors on such spaces in a form adapted to the respective N--connection
structure.

\begin{definition}
A h--spinor bundle $\ ^{h}S$ on a h--space $h\mathbf{V}$ is a complex vector
bundle with both defined action of the h--spin group $\ ^{h}Spin(V^{n})$ on
the typical fiber and an irreducible representation of the group $\ ^{h}Spin(%
\mathbf{V})\equiv Spin(h\mathbf{V})\doteq Spin(T_{x}^{\ast }h\mathbf{V}).$
The set of sections $\Gamma (\ ^{h}S)$ defines an irreducible Clifford h--module.
\end{definition}

The concept of ''d--spinors'' has been introduced for the spaces provided
with N--connection structure \cite{vspinor,vhspinor}:

\begin{definition}
\label{ddsp} A distinguished spinor (d--spinor) bundle $\mathbf{S}\doteq (\
\ ^{h}S,\ \ \ ^{v}S)$ on a N--anho\-lo\-nom\-ic manifold $\mathbf{V},$ $\ dim%
\mathbf{V}=n+m,$ is a complex vector bundle with a defined action of the
spin d--group $Spin\ \mathbf{V}\doteq Spin(V^{n})\oplus Spin(V^{m})$ with
the splitting adapted to the N--connection structure which results in an
irreducible representation $Spin(\mathbf{V})\doteq Spin(T^{\ast }\mathbf{V}%
). $ The set of sections $\Gamma (\mathbf{S})=\Gamma \ (\ \ ^{h}{S})\oplus
\Gamma (\ \ \ ^{v}{S})$ is an irreducible Clifford d--module.
\end{definition}

If we study algebras through theirs representations, we also have to
consider various algebras related by the Morita equivalence.\footnote{%
The Morita equivalence can be analyzed by applying in N--adapted form, both
on the base and fiber spaces, the consequences of the Plymen's theorem (see
Theorem 9.3 in Ref. \cite{fgbv}; in this work, we omit details of such
considerations).}

The possibility to distinguish the $Spin(n)$ (or, correspondingly $Spin(h%
\mathbf{V}),$\ $Spin(V^{n})\oplus Spin(V^{m}))$ allows us to define an
antilinear bijection $\ ^{h}J:\ \ ^{h}S\ \rightarrow \ \ ^{h}S$ (or $\ \
^{v}J:\ \ \ ^{v}S\ \rightarrow \ \ ^{v}\ S$ and $\mathbf{J}:\ \mathbf{S}\
\rightarrow \ \mathbf{S})$ with properties of type:
\begin{eqnarray}
~^{h}J(~^{h}a\psi ) &=&~^{h}\chi (~^{h}a)~^{h}J~^{h}\psi ,\mbox{ for }%
~^{h}a\in \Gamma ^{\infty }(\mathcal{C}l(h\mathbf{V}));  \notag \\
(~^{h}J~^{h}\phi |~^{h}J~^{h}\psi ) &=&(~^{h}\psi |~^{h}\phi )\mbox{ for }%
~^{h}\phi ,~^{h}\psi \in ~^{h}S.  \label{jeq}
\end{eqnarray}

The considerations presented in this Section consists the proof of:\

\begin{theorem}
\label{mr1}Any d--metric and N--con\-nec\-ti\-on structure defines naturally
the fundamental geometric objects and structures (such as the Clifford
h--module, v--module and Clifford d--modules,or the h--spin, v--spin
structures and d--spinors) for the corresponding nonholonomic spin manifold
and/or N--anholo\-nom\-ic spinor (d--spinor) manifold.
\end{theorem}

We note that similar results were obtained in Refs. \cite{vspinor,vhspinor}
for the standard Finsler and Lagrange geometries and theirs higher order
generalizations. In a more restricted form, the idea of Theorem \ref{mr1}
can be found in Ref. \cite{vncg}, where the first models of noncommutative
Finsler geometry and related gravity were analyzed.

\section{Nonholonomic Dirac Operators\newline
and Spectral Triples}

\label{s3} The Dirac operator for a class of (non) commutative nonholonomic
spaces provided with d--metric structure was introduced in Ref. \cite{vncg}
following previous constructions for the Dirac equations on locally
anisotropic spaces (various variants of Finsler--Lagrange and
Cartan--Hamilton spaces and generalizations), see \cite%
{vspinor,vhspinor,vstav,vvicol} and Part III in \cite{vsgg}. In this
Section, we define nonholonomic Dirac operators for general N--anholonomic
manifolds.

\subsection{N--anholonomic Dirac operators}

The geometric constructions depend on the type of linear connections
considered for definition of such Dirac operators. They are metric
compatible and N--adapted if the canonical d--connection is used (similar
constructions can be performed for any deformation which results in a
metric compatible d--connection).

\subsubsection{Noholonomic vielbeins and spin d--connections}

Let us consider a Hilbert space of finite dimension. For a local dual
coordinate basis $e^{\underline{i}}\doteq dx^{\underline{i}}$ on $h\mathbf{V}%
,$\ we may respectively introduce certain classes of orthonormalized
vielbeins and the N--adapted vielbeins, $\ e^{\hat{\imath}}\doteq e_{\
\underline{i}}^{\hat{\imath}}(x,y)\ e^{\underline{i}}$ and $e^{i}\doteq e_{\
\underline{i}}^{i}(x,y)\ e^{\underline{i}},$when $g^{\underline{i}\underline{%
j}}\ e_{\ \underline{i}}^{\hat{\imath}}e_{\ \underline{j}}^{\hat{\jmath}%
}=\delta ^{\hat{\imath}\hat{\jmath}}$ and $g^{\underline{i}\underline{j}}\
e_{\ \underline{i}}^{i}e_{\ \underline{j}}^{j}=g^{ij}.$

We define the algebra of Dirac's gamma horizontal matrices (in brief, gamma
h--matrices defined by self--adjoint matrices $M_{k}(\mathbb{C})$ where $%
k=2^{n/2}$ is the dimension of the irreducible representation of $\mathcal{C}%
l(h\mathbf{V})$ from relation $\ \gamma ^{\hat{\imath}}\gamma ^{\hat{\jmath}%
}+\gamma ^{\hat{\jmath}}\gamma ^{\hat{\imath}}=2\delta ^{\hat{\imath}\hat{%
\jmath}}\ ^{h}\mathbb{I}.$ \ The action of $dx^{i}\in \mathcal{C}l(h\mathbf{V%
})$ on a spinor $\ \ ^{h}\psi \in \ ^{h}S$ is given by formulas
\begin{equation}
\ \ \ ^{h}c(dx^{\hat{\imath}})\doteq \gamma ^{\hat{\imath}}\mbox{ and }\ \ \
^{h}c(dx^{i})\ \ ^{h}\psi \doteq \gamma ^{i}\ \ ^{h}\psi \equiv e_{\ \hat{%
\imath}}^{i}\ \gamma ^{\hat{\imath}}\ \ ^{h}\psi .  \label{gamfibb}
\end{equation}

Similarly, we can define the algebra of Dirac's gamma vertical matrices
related to a typical fiber $F$ (in brief, gamma v--matrices defined by
self--adjoint matrices $M_{k}^{\prime }(\mathbb{C}),$ where $k^{\prime
}=2^{m/2}$ is the dimension of the irreducible representation of $\mathcal{C}%
l(F)$) from relation $\gamma ^{\hat{a}}\gamma ^{\hat{b}}+\gamma ^{\hat{b}%
}\gamma ^{\hat{a}}=2\delta ^{\hat{a}\hat{b}}\ ^{v}\mathbb{I}.$ The action of
$dy^{a}\in \mathcal{C}l(F)$ on a spinor $\ ^{v}\psi \in \ ^{v}S$ is
\begin{equation*}
\ ^{v}c(dy^{\hat{a}})\doteq \gamma ^{\hat{a}}\mbox{ and }\ ^{v}c(dy^{a})\
^{v}\psi \doteq \gamma ^{a}\ ^{v}\psi \equiv e_{\ \hat{a}}^{a}\ \gamma ^{%
\hat{a}}\ ^{v}\psi .
\end{equation*}

A more general gamma matrix calculus with distinguished gamma matrices (in
brief, gamma d--matrices\footnote{%
in some our previous works \cite{vspinor,vhspinor} we wrote $\sigma $
instead of $\gamma $}) can be elaborated for N--anholonomic manifolds $%
\mathbf{V}$ provided with d--metric structure $\mathbf{g}=~^{h}g\oplus
~^{v}h]$ and for d--spinors $\breve{\psi}\doteq (~^{h}\psi ,\ ~^{v}\psi )\in
\mathbf{S}\doteq (~^{h}S,\ ~^{v}S).$ In this case, we consider d--gamma
matrix relations $\ \gamma ^{\hat{\alpha}}\gamma ^{\hat{\beta}}+\gamma ^{%
\hat{\beta}}\gamma ^{\hat{\alpha}}=2\delta ^{\hat{\alpha}\hat{\beta}}\mathbb{%
\ I},$ with the action of $du^{\alpha }\in \mathcal{C}l(\mathbf{V})$ on a
d--spinor $\breve{\psi}\in \ \mathbf{S}$ resulting in distinguished
irreducible representations
\begin{equation}
\mathbf{c}(du^{\hat{\alpha}})\doteq \gamma ^{\hat{\alpha}}\mbox{ and }%
\mathbf{c}=(du^{\alpha })\ \breve{\psi}\doteq \gamma ^{\alpha }\ \breve{\psi}%
\equiv e_{\ \hat{\alpha}}^{\alpha }\ \gamma ^{\hat{\alpha}}\ \breve{\psi}
\label{gamfibd}
\end{equation}%
which allows us to write
%\begin{equation*}
 $\gamma ^{\alpha }(u)\gamma ^{\beta }(u)+\gamma ^{\beta }(u)\gamma ^{\alpha
}(u)=2g^{\alpha \beta }(u)\ \mathbb{I}.$
%\end{equation*}%

In the canonical representation, we have the irreducible form $\breve{\gamma}%
\doteq ~^{h}\gamma \oplus \ ~^{v}\gamma $ and $\breve{\psi}\doteq ~^{h}\psi
\oplus ~^{v}\psi ,$ for instance, by using block type of h-- and
v--matrices. We can also write such formulas as couples of gamma and/or h--
and v--spinor objects written in N--adapted form, $\gamma ^{\alpha }\doteq
(~^{h}\gamma ^{i},~^{v}\gamma ^{a})$ and $\breve{\psi}\doteq (~^{h}\psi ,\
~^{v}\psi ).$

The spin connection $~_{S}\nabla $ for Riemannian manifolds is induced by
the Levi--Civita connection $\ ~\Gamma ,$ $~_{S}\nabla \doteq d-\frac{1}{4}\
\Gamma _{\ jk}^{i}\gamma _{i}\gamma ^{j}\ dx^{k}.$ On N--anholonomic
manifolds, spin d--connection operators $~_{\mathbf{S}}\mathbf{\nabla }$ can
be similarly constructed from any metric compatible d--connection ${\mathbf{%
\Gamma }}_{\ \beta \mu }^{\alpha }$ using the N--adapted absolute
differential $\delta $ acting, for instance, on a scalar function $f(x,y)$
in the form
\begin{equation*}
\delta f=\left( \mathbf{e}_{\nu }f\right) \delta u^{\nu }=\left( \mathbf{e}%
_{i}f\right) dx^{i}+\left( e_{a}f\right) \delta y^{a},
\end{equation*}%
for $\delta u^{\nu }=\mathbf{e}^{\nu },$ see N--elongated operators (\ref%
{dder}) and (\ref{ddif}).

\begin{definition}
The canonical spin d--connection is defined by the canonical d--connection,
\begin{equation}
~_{\mathbf{S}}\widehat{\nabla }\doteq \delta -\frac{1}{4}\ \widehat{\mathbf{%
\Gamma }}_{\ \beta \mu }^{\alpha }\gamma _{\alpha }\gamma ^{\beta }\delta
u^{\mu },  \label{csdc}
\end{equation}%
where the N--adapted coefficients $\widehat{\mathbf{\Gamma }}_{\ \beta \mu
}^{\alpha }$ are given by formulas (\ref{candcon}).
\end{definition}

We note that the canonical spin d--connection $~_{\mathbf{S}}\widehat{\nabla
}$ is metric compatible and contains nontrivial d--torsion coefficients
induced by the N--anholonomy relations (see formulas (\ref{dtors}) proved
for arbitrary d--connection).

\subsubsection{Dirac d--operators}

We consider a vector bundle $\mathbf{E}$ on a N--anholonomic manifold $%
\mathbf{V}$ (with two compatible N--connections defined as h-- and
v--splittings of $T\mathbf{E}$ and $T\mathbf{V}$)). A d--connection $%
\mathcal{D}:\ \Gamma ^{\infty }(\mathbf{E})\rightarrow \Gamma ^{\infty }(%
\mathbf{E})\otimes \Omega ^{1}(\mathbf{V})$ preserves by parallelism
splitting of the tangent total and base spaces and satisfy the Leibniz
condition $\mathcal{D}(f\sigma )=f(\mathcal{D}\sigma )+\delta f\otimes
\sigma ,$ for any $f\in C^{\infty }(\mathbf{V}),$ and $\sigma \in \Gamma
^{\infty }(\mathbf{E})$ and $\delta $ defining an N--adapted exterior
calculus by using N--elongated operators (\ref{dder}) and (\ref{ddif}) which
emphasize d--forms instead of usual forms on $\mathbf{V},$ with the
coefficients taking values in $\mathbf{E}.$

The metricity and Leibniz conditions for $\mathcal{D}$ are written
respectively
\begin{equation}
\mathbf{g}(\mathcal{D}\mathbf{X},\mathbf{Y})+\mathbf{g}(\mathbf{X},\mathcal{D%
}\mathbf{Y})=\delta \lbrack \mathbf{g}(\mathbf{X},\mathbf{Y})],  \label{mc1}
\end{equation}%
for any $\mathbf{X},\ \mathbf{Y}\in \chi (\mathbf{V}),$ and
\begin{equation}
\mathcal{D}(\sigma \beta )\doteq \mathcal{D}(\sigma )\beta +\sigma \mathcal{D%
}(\beta ),  \label{lc1}
\end{equation}%
for any $\sigma ,\beta \in \Gamma ^{\infty }(\mathbf{E}).$ For local
computations, we may define the corresponding coefficients of the geometric
d--objects and write
\begin{equation*}
\mathcal{D}\sigma _{\acute{\beta}}\doteq {\mathbf{\Gamma }}_{\ {\acute{\beta}%
}\mu }^{\acute{\alpha}}\ \sigma _{\acute{\alpha}}\otimes \delta u^{\mu }={%
\mathbf{\Gamma }}_{\ {\acute{\beta}}i}^{\acute{\alpha}}\ \sigma _{\acute{%
\alpha}}\otimes dx^{i}+{\mathbf{\Gamma }}_{\ {\acute{\beta}}a}^{\acute{\alpha%
}}\ \sigma _{\acute{\alpha}}\otimes \delta y^{a},
\end{equation*}%
where fiber ''acute'' indices are considered as spinor ones.

The respective actions of the Clifford d--algebra and Clifford h--algebra
can be transformed into maps $\Gamma ^{\infty }(\mathbf{S})\otimes \Gamma (%
\mathcal{C}l(\mathbf{V}))$ and $\Gamma ^{\infty }(~^{h}S)\otimes \Gamma (%
\mathcal{C}l(~^{h}V)$ to $\Gamma ^{\infty }(\mathbf{S})$ and, respectively, $%
\Gamma ^{\infty }(~^{h}S)$ by considering maps of type (\ref{gamfibb}) and (%
\ref{gamfibd})
\begin{equation*}
\widehat{\mathbf{c}}(\breve{\psi}\otimes \mathbf{a})\doteq \mathbf{c}(%
\mathbf{a})\breve{\psi}\mbox{\ and\ }~^{h}\widehat{c}(~^{h}{\psi }\otimes
~^{h}{a})\doteq ~^{h}{c}(~^{h}{a})~^{h}{\psi }.
\end{equation*}

\begin{definition}
\label{dddo} The Dirac d--operator (Dirac h--operator, or v--operant) on a
spin N--anholonomic manifold $(\mathbf{V},\mathbf{S},J)$ (on a h--spin
manifold\newline
$(h\mathbf{V},~^{h}S,~^{h}J),$ or on a v--spin manifold $(v\mathbf{V}%
,~^{v}S,~^{v}J))$ is defined
\begin{eqnarray}
\mathbb{D} &\doteq &-i\ (\widehat{\mathbf{c}}\circ ~_{\mathbf{S}}\nabla )
\label{ddo} \\
&=&\left( \ ^{h}\mathbb{D}=-i\ (\ ^{h}\widehat{{c}}\circ \ _{\mathbf{S}%
}^{h}\nabla ),\ \ ^{v}\mathbb{D}=-i\ (\ ^{v}\widehat{{c}}\circ \ _{\mathbf{S}%
}^{v}\nabla )\right)  \notag
\end{eqnarray}%
Such N--adapted Dirac d--operators are called canonical and denoted $%
\widehat{\mathbb{D}}=(\ ^{h}\widehat{\mathbb{D}},\ ^{v}\widehat{\mathbb{D}}\
)$\ if they are defined for the canonical d--connection (\ref{candcon}) and
respective spin d--connection (\ref{csdc}).
\end{definition}

We formulate:

\begin{theorem}
\label{mr2} Let $(\mathbf{V},\mathbf{S},\mathbf{J})$ (\ $(h\mathbf{V},\
^{h}S,\ ^{h}J))$ be a spin N--anholonomic manifold ( h--spin space). There
is the canonical Dirac d--operator (Dirac h--operator) defined by the almost
Hermitian spin d--operator
\begin{equation*}
~_{\mathbf{S}}\widehat{\nabla }:\ \Gamma ^{\infty }(\mathbf{S})\rightarrow
\Gamma ^{\infty }(\mathbf{S})\otimes \Omega ^{1}(\mathbf{V})
\end{equation*}%
(spin h--operator $\ ~_{\mathbf{S}}^{h}\widehat{\nabla }:\ \Gamma ^{\infty
}(\ ^{h}{S})\rightarrow \Gamma ^{\infty }(\ ^{h}S)\otimes \Omega ^{1}(h%
\mathbf{V})\ )$ commuting with $\mathbf{J}$ ($\ ^{h}J$), see (\ref{jeq}),
and satisfying the conditions
\begin{equation*}
(~_{\mathbf{S}}\widehat{\nabla }\breve{\psi}\ |\ \breve{\phi})\ +(\breve{\psi%
}\ |~_{\mathbf{S}}\widehat{\nabla }\breve{\phi})\ =\delta (\breve{\psi}\ |\
\breve{\phi})\
\end{equation*}%
and
\begin{equation*}
~_{\mathbf{S}}\widehat{\nabla }(\mathbf{c}(\mathbf{a})\breve{\psi})\ =%
\mathbf{c}(\widehat{\mathbf{D}}\mathbf{a})\breve{\psi}+\mathbf{c}(\mathbf{a}%
)~_{\mathbf{S}}\widehat{\nabla }\breve{\psi}
\end{equation*}%
for $\mathbf{a}\in \mathcal{C}l(\mathbf{V})$ and $\breve{\psi}\in \Gamma
^{\infty }(\mathbf{S})$
\begin{equation*}
(\ (~_{\mathbf{S}}^{h}\widehat{\nabla }\ ^{h}{\psi }|\ \ ^{h}\phi )\ +(~^{h}{%
\psi }\ |\ ~_{\mathbf{S}}^{h}\widehat{\nabla }\ \ ^{h}\phi )\ =\ ^{h}\delta
(~^{h}{\psi }\ |\ \ ^{h}\phi )\
\end{equation*}%
and $\ ~_{\mathbf{S}}^{h}\widehat{\nabla }(~^{h}c(~^{h}a)~^{h}{\psi })\
=~^{h}c(\ \ ^{h}\widehat{D}~^{h}a)~^{h}{\psi }+~^{h}c(~^{h}a)~_{\mathbf{S}%
}^{h}\widehat{\nabla }~^{h}{\psi }$ for $~^{h}a\in \mathcal{C}l(h\mathbf{V})$
and $\breve{\psi}\in \Gamma ^{\infty }(\ ^{h}{S)}$ )\ determined by the
metricity (\ref{mc1}) and Leibnitz (\ref{lc1}) conditions.
\end{theorem}

\begin{proof}
We sketch the main ideas of such Proofs. There are two possibilities:

The first one is similar to that given in Ref. \cite{fgbv}, Theorem 9.8, for
the Levi--Civita connection. We have to generalize the constructions for
d--metrics and canonical d--connections by applying N--elongated operators
for differentials and partial derivatives. The formulas have to be
distinguished into h-- and v--irreducible components. Such an approach can
be elaborated for d--connections with arbitrary torsions (it is not a
purpose of this work to consider such general constructions).

In a more particular case, we work with nonholonomic deformations of linear
connections of type (\ref{ddcon}). In such a case, there is a a second
possibility to provide a very simple proof on existence of the canonical
Dirac d--operator. The Levi--Civita connection, a metric compatible
d--connection and a corresponding distorsion tensor from (\ref{ddcon}) are
completely defined by a d--metric structure. Using the standard spin and
Dirac operator (defined by the Levi--Civita connection), we can construct a
unique nonholonomic deformation into the canonical Dirac d--operator $%
\widehat{\mathbb{D}}$ using three steps\footnote{%
inverse constructions are similar}:

1) For a given d--metric $\mathbf{g},$ we can compute $\Gamma _{\ \alpha
\beta }^{\gamma },$ then $\widehat{\mathbf{\Gamma }}_{\ \alpha \beta
}^{\gamma }$ , see formulas (\ref{candcon}), and $Z_{\ \alpha \beta
}^{\gamma },$ see formulas (\ref{cdeftc}), determining a unique deformation
of linear connections, $\ \Gamma _{\ \alpha \beta }^{\gamma }=\widehat{%
\mathbf{\Gamma }}_{\ \alpha \beta }^{\gamma }+Z_{\ \alpha \beta }^{\gamma }$
(\ref{cdeft}).

2) Introducing splitting (\ref{cdeft}) into formulas for $~_{S}\nabla $ and $%
~_{\mathbf{S}}\widehat{\nabla },$ see formula (\ref{csdc}) and related
explanations, we define a unique splitting $\ _{S}\nabla =\ _{\mathbf{S}}%
\widehat{\nabla }+\ _{Z}\widehat{\nabla },$ where $\ _{Z}\widehat{\nabla }$
is completely determined by $Z_{\ \alpha \beta }^{\gamma }$ (and, as a
consequence, by $\mathbf{g).}$ For simplicity, we omit explicit formulas for
operators $\ _{S}\nabla$ and $\ _{Z}\widehat{\nabla }.$

3) Following Definition \ref{dddo}, the privious splitting sum for spin
d--operators results in a corresponding splitting formula for the Dirac
operator, see explicitly formula (\ref{ddo}), when $\ _{S}\mathbb{D = D}+\
_{Z}\mathbb{D},$ for $_{S}\mathbb{D},$ defined by the Levi--Civita
connection, and $\ _{Z}\mathbb{D},$ induced by $Z_{\ \alpha \beta }^{\gamma
} $ (\ref{cdeftc}). For simplicity, we omit explicit formulas for operators $%
\ _{S}\mathbb{D}$ and $\ _{Z}\mathbb{D}.$ Such a splitting has an associated
splitting for the corresponding almost Hermitian spin d--operator, mentioned
in the formulation of this Theorem.

We conclude that on a spin N--anholonomic manifold we can work equivalently
with two canonical Dirac operators, $\ _{S}\mathbb{D}$ and $\mathbb{D}.$
From a formal point of view, the canonical Dirac d--operator $\mathbb{D}$
encode a nonholonomically induced torsion d--tensor, but such a distorsion
d--tensor is completely defined by a d--metric $\mathbf{g,}$ which is quite
similar to constructions with the Levi--Civita connection. $\Box $
\end{proof}

The geometric information of a spin manifold (in particular, the metric) is
contained in the Dirac operator. For nonholonomic manifolds, the canonical
Dirac d--operator has h-- and v--irreducible parts related to off--diagonal
metric terms and nonholonomic frames with associated structure. In a more
special case, the canonical Dirac d--operator is defined by the canonical
d--connection. Nonholonomic Dirac d--operators contain more information than
the usual, holonomic, ones.

\begin{proposition}
\label{pdohv} If $\widehat{\mathbb{D}}=(\ ^{h}\widehat{\mathbb{D}},\ ^{v}%
\widehat{\mathbb{D}}\ )$\ is the canonical Dirac d--operator then
\begin{eqnarray*}
\left[ \widehat{\mathbb{D}},\ f\right] &=&i\mathbf{c}(\delta f),%
\mbox{
equivalently,} \\
\left[\ ^{h}\widehat{\mathbb{D}},\ f\right] +\left[\ ^{v}\widehat{\mathbb{D}}%
,\ f\right] &=&i\ ~^{h}c(dx^{i}\frac{\delta f}{\partial x^{i}})+i~\
^{v}c(\delta y^{a}\frac{\partial f}{\partial y^{a}}),
\end{eqnarray*}%
for all $f\in C^{\infty }(\mathbf{V}).$
\end{proposition}

\begin{proof}
It is a straightforward computation following from Definition \ref{dddo}. $\Box $
\end{proof}

The canonical Dirac d--operator and its h-- and v--components have all the
properties of the usual Dirac operators (for instance, they are
self--adjoint but unbounded). It is possible to define a scalar product on $%
\Gamma ^{\infty }(\mathbf{S})$,
\begin{equation}
<\breve{\psi},\breve{\phi}>\doteq \int_{\mathbf{V}}(\breve{\psi}|\breve{\phi}%
)|\nu _{\mathbf{g}}|  \label{scprod}
\end{equation}%
where $\ \nu _{\mathbf{g}}=\sqrt{det|g|}\ \sqrt{det|h|}\ dx^{1}...dx^{n}\
dy^{n+1}...dy^{n+m}$ is the volume d--form on the N--anholonomic manifold $%
\mathbf{V}.$

\subsection{N--adapted spectral triples and\newline
distance in d--spinor spaces}

We denote $\ ^{N}\mathcal{H}\doteq L_{2}(\mathbf{V},\mathbf{S})=\left[ \ ^{h}%
\mathcal{H}=L_{2}(h\mathbf{V},\ ^{h}S),\ ^{v}\mathcal{H}=L_{2}(v\mathbf{V},\
^{v}S)\right] $ the Hilbert d--space obtained by completing $\Gamma ^{\infty
}(\mathbf{S})$ with the norm defined by the scalar product (\ref{scprod}). \
Similarly to the holonomic spaces, by using formulas (\ref{ddo}) and (\ref%
{csdc}), one may prove that there is a self--adjoint unitary endomorphism $%
~_{[cr]}\Gamma $ of $~^{N}\mathcal{H},$ called ''chirality'', being a $%
\mathbb{Z}_{2}$ graduation of $~^{N}\mathcal{H},$ \footnote{%
we use the label\ $[cr]$\ in order to avoid misunderstanding with the symbol
$\Gamma $ used for linear connections.} which satisfies the condition $%
\widehat{\mathbb{D}}\ ~_{[cr]}\Gamma =-~_{[cr]}\Gamma \ \widehat{\mathbb{D}}%
. $ %\label{chrc}
Such conditions can be written also for the irreducible components $\ ^{h}%
\widehat{\mathbb{D}}$ and$\ ^{v}\widehat{\mathbb{D}}\ .$

\begin{definition}
A distinguished canonical spectral triple (canonical spectral d--triple) $%
(~^{N}\mathcal{A},~^{N}\mathcal{H},\ \widehat{\mathbb{D}})$ for a d--algebra
$~^{N}\mathcal{A}$ is defined by a Hilbert d--space $~^{N}\mathcal{H},$ a
representation of $~^{N}\mathcal{A}$ in the algebra $~^{N}\mathcal{B}(~^{N}%
\mathcal{H})$ of d--operators bounded on $~^{N}\mathcal{H},$ and by a
self--adjoint d--operator $~^{N}\mathcal{H},$ of compact resolution,%
\footnote{%
An operator $D$ is of compact resolution if for any $\lambda \in sp(D)$ the
operator $(D-\lambda \mathbb{I})^{-1}$ is compact, see details in \cite{fgbv}%
.} such that $[~^{N}\mathcal{H},a]\in ~^{N}\mathcal{B}(~^{N}\mathcal{H})$
for any $a\in ~^{N}\mathcal{A}.$
\end{definition}

Every canonical spectral d--triple is defined by two usual spectral triples
which in our case corresponds to certain h-- and v--components induced by
the corresponding h-- and v--components of the Dirac d--operator. For such
spectral h(v)--triples we, can define the notion of $KR^{n}$--cycle and $%
KR^{m}$--cycle and consider respective Hochschild complexes. To define a
noncommutative geometry the h-- and v-- components of a canonical spectral
d--triples must satisfy certain well defined seven conditions (see Refs. \cite%
{connes1,fgbv} for details, stated there for holonomic configurations):\
 the spectral dimensions are of order $1/n$ and $1/m,$ respectively,
for h-- and v--components of the canonical Dirac d--operator;
 there are satisfied the criteria of regularity,  finiteness and reality; representations are of 1st order; there is orientability and  Poincar\'{e} duality holds true.
 Such conditions can be satisfied by any Dirac operators and canonical Dirac
d--operators (in the last case we have to work with d--objects). \footnote{%
We omit in this paper the details on axiomatics and related proofs for such considerations.}

\begin{definition}
\label{dncdg} A spectral d--triple is a real one satisfying the above mentioned seven conditions for the h-- and v--irreversible components and defining a
(d--spinor) N--anholonomic noncommutative geometry stated by the data $(\
^{N}\mathcal{A},$ $\ ^{N}\mathcal{H},\ \widehat{\mathbb{D}},\ \mathbf{J},\
~_{[cr]}\Gamma )$ and derived for the Dirac d--operator (\ref{ddo}).
\end{definition}

For N--adapted constructions, we can consider d--algebras $~^{N}\mathcal{A}%
=\ ^{h}\mathcal{A\oplus }$ $\ ^{v}\mathcal{A}$ \cite%
{vspinor,vhspinor,vstav,vvicol,vncg}. We generate N--anholonomic commutative
geometries if we take $~^{N}\mathcal{A}\doteq C^{\infty }(\mathbf{V}),$ or $%
\mathcal{~}^{h}\mathcal{A}\doteq C^{\infty }(h\mathbf{V}).$

%\subsection{Distance in d--spinor spaces}
Let us show how it is possible to compute distance in a d--spinor space:

\begin{theorem}
\label{mr3} Let $(~^{N}\mathcal{A},~^{N}\mathcal{H},\ \widehat{\mathbb{D}},%
\mathbf{J},\ _{[cr]}\Gamma )$ defines a noncommutative geometry being
irreducible for $~^{N}\mathcal{A}\doteq C^{\infty }(\mathbf{V}),$ where $%
\mathbf{V}$ is a compact, connected and oriented manifold without
boundaries, of spectral dimension $dim\ \mathbf{V}=n+m.$ In this case, there
are satisfied the conditions:

\begin{enumerate}
\item There is a unique d--metric $\mathbf{g}(\widehat{\mathbb{D}}%
)=(~^{h}g,\ ^{v}g)$ of type (\ref{m1}), \ with the ''nonlinear'' geodesic
distance on $\mathbf{V}$ defined by
\begin{equation*}
d(u_{1},u_{2})=\sup_{f\in C(\mathbf{V})}\left\{ f(u_{1},u_{2})/\parallel
\lbrack \widehat{\mathbb{D}},f]\parallel \leq 1\right\} ,
\end{equation*}%
for any smooth function $f\in C(\mathbf{V}).$

\item A N--anholonomic manifold $\mathbf{V}$\ is a spin N--anholono\-mic
space, for which the operators $\widehat{\mathbb{D}}^{\prime }$ satisfying
the condition $\mathbf{g}(\widehat{\mathbb{D}}^{\prime })=\mathbf{g}(%
\widehat{\mathbb{D}})$ define an union of affine spaces identified by the
d--spinor structures on $\mathbf{V.}$

\item The functional $\mathcal{S}(\widehat{\mathbb{D}})\doteq \int |\widehat{%
\mathbb{D}}|^{-n-m+2}$ defines a quadratic d--form with $(n+m)$--splitting
for every affine space which is minimal for $\widehat{\mathbb{D}}=%
\overleftarrow{\mathbb{D}}$ as the canonical Dirac d--operator corresponding
to the d--spin structure with the minimum proportional to the
Einstein--Hilbert action constructed for the canonical d--connection with
the d--scalar curvature $\ ^{s}\mathbf{R}$ (\ref{sdccurv}),
\begin{equation*}
\mathcal{S}(\overleftarrow{\mathbb{D}})=-\frac{n+m-2}{24}\ \int_{\mathbf{V}%
}\ ^{s}\mathbf{R}\ \sqrt{~^{h}g}\ \sqrt{~^{v}h}\ dx^{1}...dx^{n}\ \delta
y^{n+1}...\delta y^{n+m}.
\end{equation*}
\end{enumerate}
\end{theorem}

\begin{proof}
This Theorem is a generalization for N--anholonomic spaces of a similar one
formulated in Ref. \cite{connes1}, with a detailed proof presented in \cite%
{fgbv}, for the noncommutative geometry defined for a triple $(\mathcal{A},~%
\mathcal{H},\ \ _{S}\mathbb{D},\mathbf{J},\ _{[cr]}\Gamma )$\footnote{%
In the mentioned monographs, there are provided formulations/proofs of
Theorem 3.2 re--defined in terms of usual holonomic Levi--Civita structures
for Riemann/ spin manifolds}. That (holonomic) Dirac operator $\ _{S}\mathbb{%
D}$ is associated to a Levi--Civita connection and any integral with $\ \
_{S}\mathbb{D\rightarrow }D\ \ $and computed following formula \ $\int
|D|^{-n+2}\doteq \frac{1}{2^{[n/2]}{\Omega _{n}}}Wres|D|^{-n+2},$ where $%
\Omega _{n}$ is the integral of the volume on the sphere $S^{n}$ and $Wres$
is the Wodzicki residue, see details in Theorem 7.5 \cite{fgbv}. As we
sketched in the proof for Theorem\ \ref{mr2}, we get equivalent
nonholonomical configurations by distorting canonically the constructions
for the Levi--Civita connection, and related spin and Dirac operators, into
those with associated canonical d--connections, using splitting $\ \Gamma
_{\ \alpha \beta }^{\gamma }=\widehat{\mathbf{\Gamma }}_{\ \alpha \beta
}^{\gamma }+Z_{\ \alpha \beta }^{\gamma }$ (\ref{cdeft}). To such holonomic/
nonholonomic configurations, we can associate a unique distance/metric
determined in unique forms by two equivalent, holonomic and/or nonholonomic
Dirac operators. For trivial nonholonomic distorsions, the conditions and proof of this theorem for the canonical d--connection transform into those for the Levi--Civita connection \cite{connes1}. $\Box $
\end{proof}

The existence of a canonical d--connection structure which is metric
compatible and constructed from the coefficients of the d--metric and
N--connecti\-on structure is of crucial importance allowing the formulation
and proofs of the main results of this work. As a matter of principle, we
can consider any splitting of connections of type (\ref{ddcon}) and compute
a unique distance like we stated in the above Theorem \ref{mr3}, but for a
''non--canonical'' Dirac d--operator. This holds true for any noncommutative
geometry induced by a metric compatible d--connection supposed to be
uniquely induced by a metric tensor.

In more general cases, we can consider any metric compatible d--connecti\-on
with arbitrary d--torsion. Such constructions can be also elaborated in
N--adapted form by preserving the respective h- and v--irreducible
decompositions. For the Dirac d--operators, we have to start with the
Proposition \ref{pdohv} and then to repeat all constructions from \cite%
{connes1,fgbv}, both on h-- and v--subspaces. In this article, we do not
analyze (non) commutative geometries enabled with general torsions but
consider only nonholonomic deformations when distorsions are induced by a
metric structure.

Finally, we note that Theorem \ref{mr3} allows us to extract from a
canonical nonholonomic model of noncommutative geometry various types of
commutative geometries (holonomic and N--anholonomic Riemannian spaces,
Finsler--Lagrange spaces and generalizations) for corresponding nonholonomic
Dirac operators.

\section{ Spectral Functionals and Ricci Flows}

\label{s4}The goal of this section is to prove that the Perelman's
functionals \cite{gper1} and their generalizations for nonholonomic Ricci
flows in \cite{vnhrf2} [in the second reference, see formulas (29), for
commutative holonomic configurations, and (30) and (31), for commutative
nonholonomic configurations, and Theorems \ref{thmr1} and \ref{thmr2} in
this work] can be extracted from flows of a generalized Dirac operator $~^{N}%
\mathcal{D}(\chi )=~\mathbb{D}(\chi )\otimes 1$ included in spectral
functionals of type
\begin{equation}
Tr~~^{b}f(~^{N}\mathcal{D}^{2}(\chi )/\Lambda ^{2}),  \label{trperfunct}
\end{equation}%
where $~^{b}f(\chi )$ are testing functions labelled by $b=1,2,3$ and
depending on a real flow parameter $\chi ,$ which in the commutative variant
of the Ricci flow theory corresponds to that for R. Hamilton's equations %
\cite{ham1}. For simplicity, we shall use one cutoff parameter $\Lambda $
and suppose that operators under flows act on the same algebra $\mathcal{A}$
and Hilbert space $\mathcal{H},$ i.e. we consider families of spectral
triples of type $(\mathcal{A},\mathcal{H},~^{N}\mathcal{D}(\chi )).$%
\footnote{%
we shall omit in this section the left label ''N'' for algebras and Hilbert
spaces if that will not result in ambiguities}

\begin{definition}
The normalized Ricci flow equations (R. Hamilton's equations) generalized on
nonholonomic manifolds are defined in the form
\begin{equation}
\frac{\partial \mathbf{g}_{\alpha \beta }(\chi )}{\partial \chi }=-2~^{N}%
\mathbf{R}_{\alpha \beta }(\chi )+\frac{2r}{5}\mathbf{g}_{\alpha \beta
}(\chi ),  \label{normcomrf}
\end{equation}%
where $\mathbf{g}_{\alpha \beta }(\chi )$ defines a family of d--metrics
parametrized in the form (\ref{m1}) on a N-anholonomic manifold $\mathbf{V}$
enabled with a family of N--connections $N_{i}^{a}(\chi ).$
\end{definition}

The effective ''cosmological'' constant $2r/5$ in (\ref{normcomrf}) with
normalizing factor $r=\int_{v}~_{s}^{N}\mathbf{R}dv/v$ is introduced with
the aim to preserve a volume $v$ on $\mathbf{V},\,$\ where $~_{s}^{N}\mathbf{%
R}$ is the scalar curvature of\textbf{\ }type (\ref{sdccurv}), see basic
definitions and component formulas in Appendix.\footnote{%
We note that in Ref. \cite{vnhrf2} we use two mutually related flow
parameters $\chi $ and $\tau ;$ for simplicity, in this work we write only $%
\chi $ even, in general, such parameters should be rescaled for different
geometric analysis constructions.}

The corresponding family of Ricci tensors $~^{N}\mathbf{R}_{\alpha \beta
}(\chi ),$ in (\ref{normcomrf}), and nonholonomic Dirac operators $%
~^{N}D(\chi ),$ in (\ref{tract}), are defined for any value of $\chi $ by a
general metric compatible linear connection $~^{N}\mathbf{\Gamma }$ adapted
to a N--connection structure. In a particular case, we can consider the
Levi--Civita connection $~\Gamma ,$ which is used in standard geometric
approaches to physical theories. Nevertheless, for various purposes in
modelling evolution of off--diagonal Einstein metrics, constrained physical
systems, effective Finsler and Lagrange geometries, Fedosov quantization of
field theories and gravity etc\footnote{%
the coefficients of corresponding N--connection structures being defined
respectively by the generic off--diagonal metric terms, anholonomy frame
coefficients, Finsler and Lagrange fundamental functions etc}, \ it is
convenient to work with a ''N--adapted'' linear connection $~^{N}\mathbf{%
\Gamma (g).}$ If such a connection is also uniquely defined by a metric
structure $\mathbf{g},$ we are able to re--define the constructions in an
equivalent form for the corresponding Levi--Civita connection.

In noncommutative geometry, all physical information on generalized Ricci
flows can be encoded into a corresponding family of nonholonomic Dirac
operators $~^{N}\mathcal{D}(\chi ).$ For simplicity, in this work, we shall
consider that $~^{P}D=0,$ i.e. we shall not involve into the
(non)commutative Ricci flow theory the particle physics. We cite here the
work on Ricci--Yang--Mills flow \cite{streets} with evolution equations
which can be extracted from generalized spectral functionals (\ref%
{trperfunct}) with corresponding Yang--Mills nontrivial component in $%
~^{P}D. $ Perhaps a ''comprehensive'' noncommutative Ricci flow theory
should include as a stationary case the ''complete'' spectral action (\ref%
{tract}) parametrized for the standard models of gravity and particle
physics. Following such an approach, the (non)commutative/ (non)holonomic/
quantum/ classical evolution scenarios are related to topological properties
of a quantum/ classical spacetime and flows and actions of fundamental
matter fields.

\subsection{Spectral flows and Perelman functionals}

Let us consider a family of generalized d--operators%
\begin{equation}
\mathcal{D}^{2}(\chi )=-\left\{ \frac{\mathbb{I}}{2}~\mathbf{g}^{\alpha
\beta }(\chi )\left[ \mathbf{e}_{\alpha }(\chi )\mathbf{e}_{\beta }(\chi )+%
\mathbf{e}_{\beta }(\chi )\mathbf{e}_{\alpha }(\chi )\right] +\mathbf{A}%
^{\nu }(\chi )\mathbf{e}_{\nu }(\chi )+\mathbf{B}(\chi )\right\} ,
\label{oper1}
\end{equation}%
where the real flow parameter $\chi \in \lbrack 0,\chi _{0})$ and, for any
fixed values of this parameter, the matrices $\mathbf{A}^{\nu }(\chi )$ and $%
\mathbf{B}(\chi )$ are determined by a N--anholonomic Dirac operator $%
\mathbb{D}$ induced by a metric compatible d--connection $\mathbf{D,}$ see (%
\ref{dcon1}) and Definition \ref{dddo}; for the canonical d--connection, we
have to put ''hats'' on symbols and write $\widehat{\mathcal{D}}^{2},%
\widehat{\mathbf{A}}^{\nu }$ and $\widehat{\mathbf{B}}\mathbf{.}$ We
introduce two functionals $\mathcal{F}$ and $\mathcal{W}$ depending on $\chi
,$%
\begin{equation}
\mathcal{F}=Tr~\left[ ~^{1}f~(\chi )(~^{^{1}\phi }\mathcal{D}^{2}(\chi
)/\Lambda ^{2})\right] \simeq \sum\limits_{k\geq 0}~^{1}f_{(k)}(\chi
)~~^{1}a_{(k)}(~^{^{1}\phi }\mathcal{D}^{2}(\chi )/\Lambda ^{2})
\label{ncpf1a}
\end{equation}%
and
\begin{eqnarray}
\mathcal{W} &=&~^{2}\mathcal{W+}~^{3}\mathcal{W},  \label{ncpf2a} \\
\mbox{\qquad for }\ ^{e}\mathcal{W} &=&Tr\left[ \ ^{e}f(\chi )(\ ^{^{e}\phi }%
\mathcal{D}^{2}(\chi )/\Lambda ^{2})\right]  \notag \\
&=&\sum\limits_{k\geq 0}\ ^{e}f_{(k)}(\chi )\ ^{e}a_{(k)}(\ ^{^{e}\phi }%
\mathcal{D}^{2}(\chi )/\Lambda ^{2}),  \notag
\end{eqnarray}%
where we consider a cutting parameter $\Lambda ^{2}$ for both cases $e=2,3.$
Functions $~^{b}f,$ with label $b$ taking values $1,2,3,$ have to be chosen
in a form which insure that for a fixed $\chi $ we get certain compatibility
with gravity and particle physics and result in positive average energy and
entropy for Ricci flows of geometrical objects. For such testing functions,
ones hold true the formulas
\begin{eqnarray}
~^{~^{b}}f_{(0)}(\chi ) &=&\int\limits_{0}^{\infty }~^{b}f(\chi
,u)u~du,~^{b}f_{(2)}(\chi )=\int\limits_{0}^{\infty }~^{b}f(\chi ,u)~du,
\notag \\
~^{b}f_{(2k+4)}(\chi ) &=&(-1)^{k}~~^{b}f^{(k)}(\chi ,0),\quad k\geq 0.
\label{ffunct}
\end{eqnarray}%
We will comment the end of this subsection on dependence on $\chi $ of such
functions.

The coefficients $~^{b}a_{(k)}$ can be computed as the Seeley -- de Witt
coefficients \cite{gilkey} (we chose such notations when in the holonomic
case the scalar curvature is negative for spheres and the space is locally
Euclidean). In functionals (\ref{ncpf1a}) and (\ref{ncpf2a}), we consider
dynamical scaling factors of type $~^{b}\rho =\Lambda \exp (~^{~^{b}}\phi ),$
when, for instance,
\begin{eqnarray}
~^{^{1}\phi }\mathcal{D}^{2} &=&~e^{-~^{1}\phi }~\mathcal{D}^{2}e^{~^{1}\phi
}  \label{ddiracscale} \\
&=&-\left\{ \frac{\mathbb{I}}{2}~~^{^{1}\phi }\mathbf{g}^{\alpha \beta }%
\left[ ~^{^{1}\phi }\mathbf{e}_{\alpha }~^{^{1}\phi }\mathbf{e}_{\beta
}+~^{^{1}\phi }\mathbf{e}_{\beta }~^{^{1}\phi }\mathbf{e}_{\alpha }\right]
+~^{^{1}\phi }\mathbf{A}^{\nu }~^{^{1}\phi }\mathbf{e}_{\nu }+~^{^{1}\phi }%
\mathbf{B}\right\} ,  \notag
\end{eqnarray}%
\begin{eqnarray*}
\mbox{\ for \ }~^{^{1}\phi }\mathbf{A}^{\nu } &=&~e^{-2~^{1}\phi }\times
\mathbf{A}^{\nu }-2~^{^{1}\phi }\mathbf{g}^{\nu \mu }\times ~^{^{1}\phi }%
\mathbf{e}_{\beta }(^{1}\phi ), \\
~^{^{1}\phi }\mathbf{B} &\mathbf{=}&~e^{-2~^{1}\phi }\times \left( \mathbf{B-%
\mathbf{A}^{\nu }~}^{^{1}\phi }\mathbf{e}_{\beta }(^{1}\phi )\right) \mathbf{%
+~}^{^{1}\phi }\mathbf{g}^{\nu \mu }\times \mathbf{~}^{^{1}\phi }W_{\nu \mu
}^{\gamma }\mathbf{~}^{^{1}\phi }\mathbf{e}_{\gamma },
\end{eqnarray*}%
for re--scaled d--metric $~^{^{1}\phi }\mathbf{g}_{\alpha \beta
}=~e^{2~^{1}\phi }\times \mathbf{g}_{\alpha \beta }$ and N--adapted frames $%
~^{^{1}\phi }\mathbf{e}_{\alpha }=\ e^{~^{1}\phi }\times \mathbf{e}_{\alpha
} $ satisfying anholonomy relations of type (\ref{anhrel}), with re--scaled
nonholonomy coefficients $^{^{1}\phi }W_{\nu \mu }^{\gamma }.$\footnote{%
similar constructions with dilaton fields are considered in Refs. \cite%
{chcon1} and \cite{chcon3}, but in our case we work with N--anholonomic
manifolds, d-metrics and d--connections when instead ''dilatons'' there are
used scaling factors for a corresponding N--adapted Ricci flow model.} We
emphasize that similar formulas can be written by substituting respectively
the labels and scaling factors containing $\ ^{1}\phi $ with $^{2}\phi $ and
$^{3}\phi .$ For simplicity, we shall omit left labels $1,2,3$ for $\phi $
and $f,a$ if that will not result in ambiguities.

Let us denote by $~^{s}\mathbf{R}(\mathbf{g}_{\mu \nu })$ and $\mathbf{C}%
_{\mu \nu \lambda \gamma }(\mathbf{g}_{\mu \nu }),$ correspondingly, the
scalar curvature (\ref{sdccurv}) and conformal Weyl d--tensor \footnote{%
for any metric compatible d--connection $\mathbf{D,}$ the Weyl d--tensor can
be computed by formulas similar to those for the Levi--Civita connection $%
\nabla ;$ here we note that if a Weyl d--tensor is zero, in general, the
Weyl tensor for $\nabla $ does not vanish (and inversely)}%
\begin{eqnarray*}
\mathbf{C}_{\mu \nu \lambda \gamma } &=&\mathbf{R}_{\mu \nu \lambda \gamma }+%
\frac{1}{2}\left( \mathbf{R}_{\mu \lambda }\mathbf{g}_{\nu \gamma }-\mathbf{R%
}_{\nu \lambda }\mathbf{g}_{\mu \gamma }-\mathbf{R}_{\mu \gamma }\mathbf{g}%
_{\nu \lambda }+\mathbf{R}_{\nu \gamma }\mathbf{g}_{\mu \lambda }\right) \\
&&-\frac{1}{6}\left( \mathbf{g}_{\mu \lambda }\mathbf{g}_{\nu \gamma }-%
\mathbf{g}_{\nu \lambda }\mathbf{g}_{\mu \gamma }\right) ~\ _{s}\mathbf{R,}\
\end{eqnarray*}%
defined by a d--metric $\mathbf{g}_{\mu \nu }$ and a metric compatible
d--connection $\mathbf{D}$ (in our approach, $\mathbf{D}$ can be any
d--connection constructed in a unique form from $\mathbf{g}_{\mu \nu }$ and $%
\mathbf{N}_{i}^{a}$ following a well defined geometric principle). For
simplicity, we shall work on a four dimensional space and use values of type
\begin{eqnarray*}
&&\int d^{4}u~\sqrt{\det |e^{~2\phi }\mathbf{g}_{\mu \nu }|}\mathbf{R}%
(e^{2~\phi }\mathbf{g}_{\mu \nu })^{\ast ~}\mathbf{R}^{\ast }(e^{2~\phi }%
\mathbf{g}_{\mu \nu })= \\
&&\int d^{4}u~\sqrt{\det |\mathbf{g}_{\mu \nu }|}\mathbf{R}(\mathbf{g}_{\mu
\nu })^{\ast ~}\mathbf{R}^{\ast }(\mathbf{g}_{\mu \nu })= \\
&&\frac{1}{4}\int d^{4}u~\left( \sqrt{\det |\mathbf{g}_{\mu \nu }|}\right)
^{-1}\epsilon ^{\mu \nu \alpha \beta }\epsilon _{\rho \sigma \gamma \delta }%
\mathbf{R}_{\quad \mu \nu }^{\rho \sigma }\mathbf{R}_{\quad \alpha \beta
}^{\gamma \delta },
\end{eqnarray*}%
for the curvature d--tensor $\mathbf{R}_{\quad \mu \nu }^{\rho \sigma }$ (%
\ref{curv}), where sub--integral values are defined by Chern-Gauss--Bonnet
terms of type
\begin{equation*}
\mathbf{R}^{\ast }\ \mathbf{R}^{\ast }\equiv \frac{1}{4\sqrt{\det |\mathbf{g}%
_{\mu \nu }|}}\epsilon ^{\mu \nu \alpha \beta }\epsilon _{\rho \sigma \gamma
\delta }\mathbf{R}_{\mu \nu \quad }^{\ \ \rho \sigma }\mathbf{R}_{\alpha
\beta \quad }^{\ \ \gamma \delta }.
\end{equation*}

\begin{lemma}
\label{lem1}One has the four dimensional approximation
\begin{eqnarray}
&&Tr~\left[ ~f~(\chi )(~^{\phi }\mathcal{D}^{2}(\chi )/\Lambda ^{2})\right]
\simeq \frac{45}{4\pi ^{2}}~f_{(0)}\int \delta ^{4}u~e^{2\phi }\sqrt{\det |%
\mathbf{g}_{\mu \nu }|}  \label{4dapr} \\
&&+~\frac{15}{16\pi ^{2}}~f_{(2)}\int \delta ^{4}u~e^{2\phi }\sqrt{\det |%
\mathbf{g}_{\mu \nu }|}\times  \notag \\
&&\left( ~_{s}\mathbf{R}(e^{2\phi }\mathbf{g}_{\mu \nu })+3e^{-2\phi }%
\mathbf{g}^{\alpha \beta }(\mathbf{e}_{\alpha }\phi ~\mathbf{e}_{\beta }\phi
+\mathbf{e}_{\beta }\phi ~\mathbf{e}_{\alpha }\phi )\right)  \notag \\
&&+\frac{1}{128\pi ^{2}}~f_{(4)}\int \delta ^{4}u~e^{2\phi }\sqrt{\det |%
\mathbf{g}_{\mu \nu }|}\times  \notag \\
&&\left( 11~\mathbf{R}^{\ast }(e^{2\phi }\mathbf{g}_{\mu \nu })\mathbf{R}%
^{\ast }(e^{2\phi }\mathbf{g}_{\mu \nu })-18\mathbf{C}_{\mu \nu \lambda
\gamma }(e^{2\phi }\mathbf{g}_{\mu \nu })\mathbf{C}^{\mu \nu \lambda \gamma
}(e^{2\phi }\mathbf{g}_{\mu \nu })\right) .  \notag
\end{eqnarray}
\end{lemma}

\begin{proof}
It consists from a computation of $Tr$ being a N--adapted version of the
calculus provided for formula (21) in Ref. \cite{chcon3}. In our case, we do
not consider ''gauge'' fields and fermionic interactions and work with a
metric compatible d--connection and N--elongated partial derivative and
differential operators (\ref{dder}) and (\ref{ddif}). $\ $

For simplicity, we sketch the computation of coefficient before $~f_{(2)},$
i.e. the terms $\ ^{e}a_{(2)}(\ ^{^{e}\phi }\mathcal{D}^{2}(\chi )/\Lambda
^{2})$ in functionals (\ref{ncpf1a}) and (\ref{ncpf2a}), omitting the left
label ''$e"$ and writing $a_{(2)}$ for a corresponding $f_{(2)}.$ Following
Theorem 4.8 from Ref. \cite{gilkey} (proof of that theorem can be
generalized for an arbitrary metric compatible d--connection), we have that
the so--called second Seeley-- de Witt coefficient is
\begin{equation*}
a_{(2)}(\mathcal{D}^{2}(\chi )/\Lambda ^{2})=\frac{\Lambda ^{2}}{16\pi ^{2}}%
\int_{\mathbf{V}}\delta V\ Tr\left( -\frac{\ _{s}\mathbf{R}}{6}\right) .
\end{equation*}%
This coefficient can be used for evaluating $\ a_{(2)}(\ ^{\phi }\mathcal{D}%
^{2}(\chi )/\Lambda ^{2})$ following the method of conformal transforms for
operators and functionals contained in spectral actions which was developed
in sections II and III of Ref. \cite{chcon3}. We have only to perform a
similar calculus on N--anholonomic manifolds using d--connections and
N--adapted frames/operators.

For a conformally transformed inverse d--metric $\ ^{\phi }\mathbf{g,}$ when
$\ ^{\phi }\mathbf{g}^{\mu \nu }=e^{-2\phi }\mathbf{g}^{\mu \nu },$ $\mathbf{%
g}^{\mu \nu }$ being inverse to coefficients of $\mathbf{g}_{\mu \nu }$ (\ref%
{m1}), we have the formula for conformal transform of scalar of curvature (%
\ref{sdccurv}),%
\begin{equation*}
\ _{s}\mathbf{R(^{\phi }\mathbf{g})=}e^{-2\phi }\left[ \ _{s}\mathbf{R(%
\mathbf{g})}+3\mathbf{g}^{\mu \nu }\left( \ ^{g}\mathbf{D}_{\mu }\ ^{g}%
\mathbf{D}_{\nu }+\ ^{g}\mathbf{D}_{\nu }\ ^{g}\mathbf{D}_{\mu }+\mathbf{e}%
_{\mu }\mathbf{e}_{\nu }+\mathbf{e}_{\nu }\mathbf{e}_{\mu }\right) \phi %
\right] ,
\end{equation*}%
where $\ ^{g}\mathbf{D}_{\mu }$ is a metric compatible d--connection
completely determined by $\mathbf{g.}$ Using the identity
\begin{equation*}
a_{(2)}\left( u,e^{-\phi }\mathcal{D}^{2}e^{-\phi }\right) =a_{(2)}\left( u,%
\mathcal{D}^{2}e^{-2\phi }\right) =a_{(2)}\left( u,e^{-2\phi }\mathcal{D}%
^{2}\right) ,
\end{equation*}%
which can be verified by straightforward computations with operator $%
\mathcal{D}^{2}$ (\ref{oper1}) containing N--adapted derivatives $\mathbf{e}%
_{\mu }\ $(\ref{dder}), and putting together all terms we get that
\begin{eqnarray*}
\ a_{(2)}(\ ^{\phi }\mathcal{D}^{2}(\chi )/\Lambda ^{2})&=& \frac{15}{16\pi
^{2}}\int_{\mathbf{V}}\delta Ve^{2\phi } \times \\
&& \left[ ~_{s}\mathbf{R}(\ ^{\phi }\mathbf{g})+3\ ^{\phi }\mathbf{g}%
^{\alpha \beta }(\mathbf{e}_{\alpha }\phi ~\mathbf{e}_{\beta }\phi +\mathbf{e%
}_{\beta }\phi ~\mathbf{e}_{\alpha }\phi )\right] ,
\end{eqnarray*}%
i.e. the coefficient before $~f_{(2)}$ in (\ref{4dapr}).

Finally, we note that generalizing the calculus from \cite{chcon3} for
d--connections and N--adapted frames, we can similarly compute the
coefficients\newline
$\ a_{(0)}(\ ^{\phi }\mathcal{D}^{2}(\chi )/\Lambda ^{2})$ and $a_{(2)}(\
^{\phi }\mathcal{D}^{2}(\chi )/\Lambda ^{2}),$ for any chosen conformal
transform $\phi $ (in general, with labels $\ ^{b}\phi )$ \ and parameter $%
\chi .$ Summarizing all necessary terms, we get the approximation (\ref%
{4dapr}) . $\Box $
\end{proof}

\vskip4pt Let us state some additional hypotheses which will be used for
proofs of the theorems in this section: Hereafter we shall consider a four
dimensional compact N--anholonomic manifold $\mathbf{V,}$ with volume forms $%
\delta V=\sqrt{\det |\mathbf{g}_{\mu \nu }|}\delta ^{4}u$ and normalization $%
\int\nolimits_{\mathbf{V}}\delta V~\mu =1$ for $\mu =e^{-f}(4\pi \chi
)^{-(n+m)/2}$ with $f$ being a scalar function $f(\chi ,u)$ and $\chi >0.$

Now, we are able to formulate the main results of this Paper:

\begin{theorem}
\label{thmr1}For the scaling factor $~^{1}\phi =-f/2,$ the spectral
functional (\ref{ncpf1a}) can be approximated $\mathcal{F}=\ ^{P}\mathcal{F}(%
\mathbf{g,D,}f),$ where the first Perelman functional (in our case for
N--anholonomic Ricci flows) is
\begin{equation*}
\ ^{P}\mathcal{F}=\int\nolimits_{\mathbf{V}}\delta V~e^{-f}\left[ ~_{s}%
\mathbf{R}(e^{-f}\mathbf{g}_{\mu \nu })+\frac{3}{2}e^{f}\mathbf{g}^{\alpha
\beta }(\mathbf{e}_{\alpha }f~\mathbf{e}_{\beta }f+\mathbf{e}_{\beta }f~%
\mathbf{e}_{\alpha }f)\right] .
\end{equation*}
\end{theorem}

\begin{proof}
We introduce $~^{1}\phi =-f/2$ into formula (\ref{4dapr}) from Lemma \ref%
{lem1}. We can rescale the flow parameter $\chi \rightarrow \check{\chi}$
such way that $\frac{3}{2}\exp [f(\chi )]=\exp [\check{f}(\check{\chi})].$
We get that up to a scaling factor $\check{f}$ and additional fixing of a
new test function to have the coefficients $~~^{1}\check{f}_{(2)}=16\pi
^{2}/15$ and $~~^{1}\check{f}_{(0)}=~^{1}\check{f}_{(4)}=0,$ computed for
''inverse hat'' values by choosing necessary values of $\check{f}$ and $%
\check{\chi}$ in formulas (\ref{ffunct}), the value
\begin{equation*}
^{P}\mathcal{F}\sim \int\nolimits_{\mathbf{V}}\delta V~e^{-\check{f}}(~_{s}%
\mathbf{R}+|\mathbf{D}\check{f}|^{2})
\end{equation*}%
is just the N--anholonomic version of the first Perelman functional (formula
(30) in Ref. \cite{vnhrf2}). Taking $\mathbf{D=\nabla ,}$ we get the
well--known formula for Ricci flows of Riemannian metrics \cite{gper1}.

One should be noted here that the coefficient $\frac{3}{2}$ was re--scaled
by imposing a corresponding nonholonomic constraint on functionals under
consideration, which is possible for Ricci flows (such re--definitions of
flow parameters were considered in Perelman's work \cite{gper1}; additional
nonholonomic constraints and evolutions being introduced in \cite{vnhrf2}).
Such ''re--scaled'' approximations are not possible if we extract certain
commutative physical models from noncommutative spectral actions (i.e. not
from evolution functionals) like in Refs. \cite%
{chcon2,chamconmarc,chcon3,cmr,ccm}. If we fix from the very beginning a
Ricci flow parameter $\chi $ (not allowing re--scalings), we have to correct
the resulting (non) holonomic Perelman like functionals by introducing
certain additional coefficients like $3/2$ etc, which can be interpreted as
some contributions from noncommutative geometry for certain evolution
models. $\Box $
\end{proof}

Sketching the proof of the above theorem and further theorems in this
section, we can use the techniques elaborated in Refs. \cite{gper1} but
generalized for functionals depending on a flow parameter and performing
necessary approximations on N--anholonomic manifolds. There are some
important remarks.

\begin{remark}
\label{r15}For nonholonomic Ricci flows of (non)commutative geometries, we
have to adapt the evolution to certain N--connection structures (i.e.
nonholonomic constraints). This results in  additional possibilities to
re--scale coefficients and parameters in spectral functionals and their
commutative limits:

\begin{enumerate}
\item The evolution parameter $\chi ,$ scaling factors $\ ^{b}f$ and
nonholonomic constraints and coordinates can be re--scaled/ redefined (for
instance, $\chi \rightarrow \check{\chi}$ and $\ ^{b}f\rightarrow \ ^{b}%
\check{f})$ such a way that the spectral functionals have limits to some 'standard'
nonholonomic versions of Perelman functionals (with prescribed types of  coefficients) considered
in Ref. \cite{vnhrf2}.

\item Using additional dependencies on $\chi $ and freedom in choosing
scaling factors $\ ^{b}f(\chi ),$ we can prescribe such nonholonomic
constraints/ configurations on evolution equations (for instance, with $^{1}%
\check{f}_{(2)}=16\pi ^{2}/15$ and $~~^{1}\check{f}_{(0)}=~^{1}\check{f}%
_{(4)}=0)$ when the spectral functionals result exactly in necessary types
of effective Perelman \ functionals (with are commutative, but, in general,
nonholonomic).

\item For simplicity, we shall write in brief only $\chi $ and $f$
considering that we have chosen such scales, parametrizations of coordinates
and N--adapted frames and flow parameters when coefficients in spectral
functionals and resulting evolution equations maximally correspond to
certain generally accepted commutative physical actions/ functionals.

\item For nonholonomic Ricci flow models (commutative or noncommutative
ones) with a fixed evolution parameter $\chi ,$ we can construct certain
effective nonholonomic evolution models with induced noncommutative
corrections for coefficients.

\item Deriving effective nonholonomic evolution models from spectral
functionals, we can use the technique of ''extracting'' physical models from
spectral actions, elaborated in \cite{chcon2,chamconmarc,chcon3,cmr,ccm}, see also
references therein. For commutative and/or noncommutative geometric/
physical models of nonholonomic Ricci flows, we have to generalize the
approach to include spectral functionals and N--adapted evolution equations
depending on the type of nonholonomic constraints, normalizations and
re--scalings of constants and effective conformal factors.
\end{enumerate}
\end{remark}

\vskip4ptWe ''extract'' from the second spectral functional (\ref{ncpf2a})
another very important physical value:

\begin{theorem}
\label{thmr2}The functional (\ref{ncpf2a}) is approximated $\mathcal{W}=~^{P}%
\mathcal{W}(\mathbf{g,D,}f,\chi ),$ where the second Perelman functional is%
\begin{eqnarray*}
&&\ ^{P}\mathcal{W}=\int\nolimits_{\mathbf{V}}\delta V~\mu \times \\
&&\left[ \chi \left( ~_{s}\mathbf{R}(e^{-f}\mathbf{g}_{\mu \nu })+\frac{3}{2}%
e^{f}\mathbf{g}^{\alpha \beta }(\mathbf{e}_{\alpha }f~\mathbf{e}_{\beta }f+%
\mathbf{e}_{\beta }f~\mathbf{e}_{\alpha }f)\right) +f-(n+m)\right] ,
\end{eqnarray*}
for scaling $~^{2}\phi =-f/2$ in $~^{2}\mathcal{W}$ \ and $~^{3}\phi =(\ln
|f-(n+m)|-f)/2$ in $~^{3}\mathcal{W},$ \ from (\ref{ncpf2a}).
\end{theorem}

\begin{proof}
Let us compute $~\mathcal{W=}~^{2}\mathcal{W+}~^{3}\mathcal{W},$ using
formula (\ref{4dapr}), for $~~^{2}\mathcal{W}$ defined by $~^{2}\phi =-f/2$ $%
\ $ with $~^{2}f_{(0)}(\chi )=~^{2}f_{(4)}(\chi )=0$ and $~^{2}f_{(2)}(\chi
)=16\pi ^{2}/[15(4\pi \chi )^{(n+m)/2}]$ and $~~^{3}\mathcal{W}$ defined by $%
~^{3}\phi =(\ln |f-(n+m)|-f)/2$ with $^{3}f_{(2)}(\chi )=~^{3}f_{(4)}(\chi
)=0$ and $~^{3}f_{(0)}(\chi )=4\pi ^{2}/[45$ $(4\pi \chi )^{(n+m)/2}].$ The
possibility to use parametrizations of scaling factors and imposed types of
nonholonomic constraints on evolution functionals follows from Remark %
\ref{r15} and, in this case, the approximations are similar to those
performed in the proof of Theorem \ref{thmr1}. After a corresponding
redefinition of coordinates, we get
\begin{equation*}
~^{P}\mathcal{W}\sim \int\nolimits_{\mathbf{V}}\delta V~\mu ~[\chi (~_{s}%
\mathbf{R}+|\mathbf{D}f|^{2})+f-(n+m)]
\end{equation*}%
which is just the N--anholonomic version of the second Perelman functional
(formula (31) in Ref. \cite{vnhrf2}). Taking $\mathbf{D=\nabla ,}$ we obtain
a formula for Ricci flows of Riemannian metrics \cite{gper1}. $\Box $
\end{proof}

\vskip4pt The nonholonomic version of Hamilton equations (\ref{normcomrf})
can be derived from commutative Perelman functionals $\ ^{P}\mathcal{F}$ and
$~^{P}\mathcal{W},$ see Theorems 3.1 and 4.1 in Ref. \cite{vnhrf2}. The
original Hamilton--Perelman Ricci flows constructions can be generated for $%
\mathbf{D=\nabla .}$ The surprising result is that even we start with a
Levi--Civita linear connection, the nonholonomic evolution will result
almost sure in generalized geometric configurations with various $\mathbf{N}$
and $\mathbf{D}$ structures.

\subsection{Spectral functionals for thermodynamical values}

Certain important thermodynamical values such as the average energy and
entropy can be derived directly from noncommutative spectral functionals as
respective commutative configurations of spectral functionals of type (\ref%
{ncpf1a}) and (\ref{ncpf2a}) but with different testing functions than in
Theorems \ref{thmr1} and \ref{thmr2}.

\begin{theorem}
\label{thae}Using a scaling factor of type $~^{1}\phi =-f/2,$ we extract
from the spectral functional (\ref{ncpf1a}) a nonholonomic version of
average energy, $\mathcal{F}\rightarrow <\mathcal{E}>,$ where
\begin{equation}
<\mathcal{E}>=-\chi ^{2}\int\nolimits_{\mathbf{V}}\delta V~\mu \left[ _{s}%
\mathbf{R}(e^{-f}\mathbf{g}_{\mu \nu })+\frac{3}{2}\mathbf{g}^{\alpha \beta
}(\mathbf{e}_{\alpha }f~\mathbf{e}_{\beta }f+\mathbf{e}_{\beta }f~\mathbf{e}%
_{\alpha }f)-\frac{n+m}{2\chi }\right]   \label{naen}
\end{equation}%
if the testing function is chosen to satisfy the conditions $%
~^{1}f_{(0)}(\chi )=4\pi ^{2}(n+m)\chi /45(4\pi \chi )^{(n+m)/2},$ $%
~^{1}f_{(2)}(\chi )=16\pi ^{2}\chi ^{2}/15(4\pi \chi )^{(n+m)/2}$ and $%
~^{1}f_{(4)}(\chi )=0.$
\end{theorem}

\begin{proof}
It is similar to that for Theorem \ref{thmr1}, but for different
coefficients of the testing function. Here, we note that, in general, the
statement of this theorem if for a different parametrization of $\chi $ and $%
f,$ see point 3 in Remark \ref{r15}. Re--defining coordinates and
nonholonomic constraints, we can write (\ref{naen}) in the form
\begin{equation*}
<\widehat{E}>\sim -\chi ^{2}\int\nolimits_{\mathbf{V}}\delta V~\mu (~_{s}%
\mathbf{R}+|\mathbf{D}f|^{2}-\frac{n+m}{2\chi })
\end{equation*}%
which is the N--anholonomic version of average energy from Theorem 4.2 in
Ref. \cite{vnhrf2}). We get the average energy for Ricci flows of Riemannian
metrics \cite{gper1} if $\mathbf{D=\nabla .}$ $\Box $
\end{proof}

\vskip4pt

Similarly to Theorem \ref{thmr2} (inverting the sign of nontrivial
coefficients of the testing function) we prove:

\begin{theorem}
\label{thnhs}We extract  a nonholonomic version of entropy of nonholonomic
Ricci flows from the functional (\ref{ncpf2a}), $\mathcal{W}\rightarrow ~%
\mathcal{S},$ where
\begin{eqnarray*}
&&\mathcal{S}=-\int\nolimits_{\mathbf{V}}\delta V~\mu \times  \\
&&\left[ \chi \left( ~_{s}\mathbf{R}(e^{-f}\mathbf{g}_{\mu \nu })-\frac{3}{2}%
e^{f}\mathbf{g}^{\alpha \beta }(\mathbf{e}_{\alpha }f~\mathbf{e}_{\beta }f+%
\mathbf{e}_{\beta }f~\mathbf{e}_{\alpha }f)\right) +f-(n+m)\right] ,
\end{eqnarray*}%
if we introduce $\delta V=\delta ^{4}u$ and $\mu =e^{-f}(4\pi \chi
)^{-(n+m)/2}$ into formula (\ref{4dapr}), for $\chi >0$ and $\int\nolimits_{%
\mathbf{V}}dV~\mu =1$ in (\ref{4dapr}), for scaling $~^{2}\phi =-f/2$ in $%
~^{2}\mathcal{W}$ \ and $~^{3}\phi =(\ln |f-(n+m)|-f)/2$ in $~^{3}\mathcal{W}%
,$ \ from (\ref{ncpf2a}).
\end{theorem}

\begin{proof}
This Theorem is a ''thermodynamic'' analog of Theorem \ref{thmr2}, in
general, with different parameterizations of the evolution parameter and
scaling factor (as we noted in\ points 3--5 of Remark \ref{r15}).  We
compute $~\mathcal{S=}~^{2}\mathcal{W+}~^{3}\mathcal{W},$ using formula (\ref%
{4dapr}), for $~~^{2}\mathcal{W}$ defined by $~^{2}\phi =-f/2$ $\ $ with $%
~^{2}f_{(0)}(\chi )=~^{2}f_{(4)}(\chi )=0$ and $~^{2}f_{(2)}(\chi )=-16\pi
^{2}/[15(4\pi \chi )^{(n+m)/2}]$ and $~~^{3}\mathcal{W}$ defined by $%
~^{3}\phi =(\ln |f-(n+m)|-f)/2$ with $^{3}f_{(2)}(\chi )=~^{3}f_{(4)}(\chi
)=0$ and $~^{3}f_{(0)}(\chi )=-4\pi ^{2}/[45$ $(4\pi \chi )^{(n+m)/2}].$
After corresponding re--parametrization and re--definition of scaling factor and
redefinition of N--adapted frames/ nonholonomic constraints, we transform $%
\mathcal{S}$ into
\begin{equation*}
~\widehat{S}\sim \int\nolimits_{\mathbf{V}}\delta V~\mu ~[\chi (~_{s}\mathbf{%
R}+|\mathbf{D}f|^{2})+f-(n+m)],
\end{equation*}%
i.e. we obtain the N--anholonomic version of Perelman's entropy, see Theorem
4.2 in Ref. \cite{vnhrf2}). For $\mathbf{D=\nabla ,}$ we get the
corresponding formula for the entropy Ricci flows of Riemannian metrics \cite%
{gper1}. $\Box $
\end{proof}

\vskip4pt We can formulate and prove a Theorem alternative to Theorem \ref%
{thae} and get the formula (\ref{naen}) from the spectral functional $~^{2}%
\mathcal{W+}~^{3}\mathcal{W}.$  Such a proof is similar to that for Theorem %
\ref{thmr2}, but with corresponding nontrivial coefficients for two testing
functions $~~^{2}f(\chi )$ and $~~^{3}f(\chi ).$ The main difference is that
for Theorem \ref{thae} it is enough to use only one testing function. We do
not present such computations in this work.

It is not surprising that certain 'commutative' thermodynamical physical
values can be derived alternatively from different spectral functionals
because such type 'commutative' thermodynamical values can be generated by a
partition function
\begin{equation}
\widehat{Z}=\exp \left\{ \int\nolimits_{\mathbf{V}}\delta V~\mu \left[ -f+%
\frac{n+m}{2}\right] \right\} ,  \label{nhpf}
\end{equation}%
associated to any $Z=\int \exp (-\beta E)d\omega (E)$ being the partition
function for a canonical ensemble at temperature $\beta ^{-1},$ which in it
turn is defined by the measure taken to be the density of states $\omega
(E). $ In this case, we can compute the average energy, $<E>=-\partial \log
Z/\partial \beta ,$ the entropy $S=\beta <E>+\log Z$ and the fluctuation $%
\sigma =<(E-<E>)^{2}>=\partial ^{2}\log Z/\partial \beta ^{2}.$

\begin{remark}
\ Following a straightforward computation for (\ref{nhpf}) (similarly to
constructions from \cite{gper1}, but following a N--adapted calculus, see
Theorem 4.2 in Ref. \cite{vnhrf2}\footnote{%
we emphasize that in this section we follow a different system of
denotations for the Ricci flow parameter and normalizing functions}) we
prove that
\begin{equation}
\widehat{\sigma }=2\chi ^{2}\int\nolimits_{\mathbf{V}}\delta V~\mu \left[
\left| R_{ij}+D_{i}D_{j}f-\frac{1}{2\chi }g_{ij}\right| ^{2}+\left|
R_{ab}+D_{a}D_{b}f-\frac{1}{2\chi }g_{ab}\right| ^{2}\right] .
\label{nhfluct}
\end{equation}
\end{remark}

Using formula $\mathbf{R}_{\mu \nu }^{2}\mathbf{=}\frac{1}{2}\mathbf{C}_{\mu
\nu \rho \sigma }^{2}-\frac{1}{2}\mathbf{R}^{\ast ~}\mathbf{R}^{\ast }+\frac{%
1}{3}\ _{s}\mathbf{R}^{2}$ (it holds true for any metric compatible
d--connections, similarly to the formula for the Levi--Civita connection;
see, for instance, Ref. \cite{chcon3}), we expect that the formula for
fluctuations (\ref{nhfluct}) can be generated directly, by corresponding
re-scalings, from a spectral action with nontrivial coefficients for testing
functions when $~f_{(4)}\neq 0,$ see formula (\ref{4dapr}). Here we note
that in the original Perelman's functionals there were not introduced terms
being quadratic on curvature/ Weyl / Ricci tensors. For nonzero $~f_{(4)},$
such terms (see Lemma \ref{lem1}) may be treated as certain noncommutative /
quantum contributions to the classical commutative Ricci flow theory. For
simplicity, we omit such considerations in this work.

The framework of Perelman's functionals and generalizations to corresponding
spectral functionals can be positively applied for developing statistical
analogies of (non) commutative Ricci flows. For instance, the functional $%
\mathcal{W}$ is the ''opposite sign'' entropy, see formulas from Theorems %
\ref{thmr2} and \ref{thnhs}. Such constructions may be considered for a
study of optimal ''topological'' configurations and evolution of both
commutative and noncommutative geometries and relevant theories of physical
interactions.

Here, one should be emphasized that the formalism of Perelman functionals
and associated thermodynamical values can not be related directly to similar
concepts in black hole physics (as it is discussed in \cite{gper1,vnhrf2})
or to quantum mechanical systems as generalized Bost--Connes systems \cite%
{cmr,ccm}. The approach is not related directly to alternative constructions
in geometric and nonequilibrium thermodynamics, locally anisotropic kinetics
and stochastic processes \cite{rup,mrug,salb,vap} for which the nonholonomic
geometric methods play an important role. Nevertheless, spectral functional
constructions seem to be important for certain noncommutative versions of
stochastic processes and kinetics of particles in constrained phase spaces
and for noncommutative mechanics models.

\section{Discussion and Conclusions}

\label{s5} To summarize, we have shown that an extension of the spectral
action formalism to spectral functionals with nonholonomic Dirac operators
includes naturally the Ricci flow theory and gravitational field equations
and various types of generalized geometric configurations modelled by
nonholonomic frames and deformations of linear connections. This unification
of the spectral triple approach to noncommutative geometry \cite%
{connes1,chamconmarc,cmr} with the Hamilton--Perelman Ricci flow theory \cite%
{ham1,gper1}, with certain new applications in physics \cite{vnhrf2,vrf},
emphasizes new advantages obtained previously following the nonlinear
connection formalism and anholonomic frame method, elaborated for standard
models of physics in Refs. \cite{vrfg,vncg,vsgg,vvicol,vpla,vlgr,vijmmp}.

We conclude that the paradigm of spectral action and spectral functionals
with nonholonomic Dirac operators is a very general one containing various
types of locally anisotropic, noncommutative, nonsymmetric spacetime
geometries and that all the correct features of the standard physical
interactions and evolution models are obtained. Such results support the
idea that all geometric and physical information about spacetime, physical
fields and evolution scenarios can be extracted from a corresponding
generalized Dirac operator, and its flows and/or stationary configurations,
on appropriated noncommutative spaces.

Let us outline the some important motivations for a systematic approach to
noncommutative Ricci flow theory provided by certain directions in modern
particle and mathematical physics.

\begin{enumerate}
\item The theory of Ricci flows with nonholonomic constraints:

In a series of papers on Ricci flows and exact solutions in gravity \cite%
{vnhrf2,vrf,vv1,vrf5}, we proved that if the evolution (Hamilton's)
equations are subjected to nonholonomic constraints the Riemannian metrics
and connections positively transform into geometric objects defining
generalized Lagrange--Finsler, nonsymmetric, noncommutative and various
other spaces.

\item The theory of spinors on Riemann--Finsler spaces:

Finsler geometry is not only a straightforward generalization of the concept
of Riemannian space to nonlinear metric elements on tangent bundle. There
were developed a set of new geometric constructions with nonlinear
connection structures and by introducing the concept of nonholonomic
manifold. It is well known that the first example of (later called) Finsler
metric is contained in the famous B. Riemann thesis from 1856, where, for
simplicity, the considerations were restricted only to quadratic forms, see
historical remarks and reviews in Refs. \cite{ma,mhohs,bejancu,bcs,bm} and %
\cite{vrf,vsgg}, on application of Lagrange--Finsler methods to standard
models of physics. But real physical nonlinear phenomena can not be
restricted only to quadratic metric elements and linear connections. It was
a very difficult task to define spinors and write the Dirac equation on
Finsler--Lagrange spaces (and generalizations) working, for instance, with
the Cartan--Finsler canonical nonlinear and linear connections, see results
outlined in Refs. \cite{vspinor,vhspinor,vstav,vvicol}. Having defined the
Dirac--Finsler/Lagrange operators, induced by the canonical distinguished
connection, it was not a problem to construct noncommutative versions of
spaces with generic local anisotropy (for instance, different models of
noncommutative Riemann--Finsler geometry, noncommutative geometric mechanics,
the constructions are summarized in the Part III of monograph \cite{vsgg}).

\item String theory and gauge gravity models:

Effective locally anisotropic (super) gravity models were derived in low
energy limits of string/M--theory, see Ref. \cite{vstrf,vncsup}. The
so--called absolute anti--symmetric torsion is a ''source'' for
noncommutative coordinate relations in such theories; on such nonholonomic
configurations, see Chapters 13 and 14 in monograph \cite{vsgg}. Here we
note that noncommutative gauge gravity models can be generated by applying
the Seiberg--Witten transform \cite{seiberg} to gauge gravity theories \cite%
{vncgr,vncg,caciatori,moffat,castro1,jurco,aschieri}. Beta functions and
renormalization problems in such theories result, in general, in
nonholonomic and noncommutative Ricci flow evolution equations.

\item Exact solutions with generic off--diagonal metrics and nonholonomic variables in gravity:

There were constructed and analyzed a number of exact solutions in modern
gravity theories, see reviews of results and references in \cite%
{vijmmp,vrfg,vsgg}, following the idea that considering nonholonomic
distributions on a commutative Einstein manifold, defined by nonholonomic
moving frames, it is possible to model Finsler like structures and
generalizations in Einstein/ string/ gauge ... gravity theory. Such
constructions are not for vector/tangent bundles, but for the (pseudo)
Riemannian/Einstein and Riemann--Cartan manifolds with local fibered
structure. Geometrically, a nonholonomic structure induces a formal torsion
even on (pseudo) Riemannian manifolds. In such cases, it is possible to work
equivalently both with the Levi--Civita and the Cartan connection, or other
metric linear connection structures completely defined by a metric. For the
Levi--Civita case, the torsion is zero, but in other cases the effective
torsions are induced by certain off-diagonal coefficients of the metric, via
nonholonomic deformations. Constructing noncommutative analogs of exact
off-diagonal solutions in different models of gravity, one obtains
noncommutative models of Finser geometries and generalizations. We emphasize
that the approach can be elaborated for standard commutative and
noncommutative models in physics and related to solutions of the Ricci flow
theory with nonholonomic variables, see \cite{vrf5,vrf,vv1} and references
therein.

\item Fedosov quantization of Einsein gravity and quantum
Lagange--Finsler spaces:

In a series of recent works, see \cite{vpla,esv,vlgr} and discussed there
reference, it was proved that the Einstein gravity can be alternatively
described in the so-called Finsler-Lagrange and almost K\"{a}hler variables
(similarly, there are equivalent formulations of the general relativity
theory in spinor, tetradic, differential forms, tensorial form etc) and
quantized following the methods of deformation quantization. Applying to
nonholonomic (pseudo) Riemannian manifolds the geometric technique developed
by Fedosov for deformation quantization, we proved that the Einstein,
Lagrange--Finsler, Hamilton--Cartan and generalized spaces can be
quantized following such methods. Using the corresponding nonholonomic Dirac
operators and spinor structures, it is possible to define generalized
Finsler like spectral triples and to define noncommutative Fedosov--Einstein,
Fedosov--Finsler etc spaces which for corresponding special cases result in
already quantized (in the meaning of deformation quantization) geometries
and their Ricci flows.
\end{enumerate}

As future directions, it might be worthwhile to pursue the results of this
paper for computing noncommutative Ricci flow corrections to physically
valuable exact solutions in gravity and elaborating noncommutative versions
of quantum gravity models in almost K\"{a}hler variables quantized following
Fedosov methods.

\vskip3pt

\textbf{Acknowledgement: } The work was partially  performed during a visit volunteer research work at Fields Institute. Author is grateful to referee for  very useful critics and important suggestions.

\appendix

\setcounter{equation}{0} \renewcommand{\theequation}
{A.\arabic{equation}} \setcounter{subsection}{0}
\renewcommand{\thesubsection}
{A.\arabic{subsection}}

\section{N--adapted Linear Connections}

The class of linear connection on a N--anholonomic manifolds splits into two
subclasses of those which are adapted or not to a given N--connection
structure.

A distinguished connection\ (d--connection, or N--adapted linear
connection)\ $\mathbf{D}$ on a N--anho\-lo\-no\-mic manifold $\mathbf{V}$ is
a linear connection conserving under parallelism the Whitney sum (\ref%
{whitney}). For any d--vector $\mathbf{X,}$ there is a decomposition of $%
\mathbf{D}$ into h-- and v--covariant derivatives,%
\begin{equation*}
\mathbf{D}_{\mathbf{X}}\mathbf{\doteqdot X}\rfloor \mathbf{D=}\ hX\rfloor
\mathbf{D+}\ vX\rfloor \mathbf{D=}Dh_{X}+D_{vX}=hD_{X}+vD_{X},
\label{dconcov}
\end{equation*}%
where the symbol ''$\rfloor "$ denotes the interior product. We shall write
conventionally that $\mathbf{D=}(hD,\ vD),$ or $\mathbf{D}_{\alpha
}=(D_{i},D_{a}).$ With respect to N--adapted bases (\ref{dder}) and (\ref%
{ddif}), the local formulas for d--connections a parametrized in the form: $%
\mathbf{D=\{\Gamma }_{\ \alpha \beta }^{\gamma }=\left(
L_{jk}^{i},L_{bk}^{a},C_{jc}^{i},C_{bc}^{a}\right) \},$ with $%
hD=(L_{jk}^{i},L_{bk}^{a})$ and $vD=(C_{jc}^{i},C_{bc}^{a}).$

The N--adapted components $\mathbf{\Gamma }_{\ \beta \gamma }^{\alpha }$ of
a d--connection $\mathbf{D}_{\alpha }=(\mathbf{e}_{\alpha }\rfloor \mathbf{D}%
),$ where ''$\rfloor $'' denotes the interior product, are computed
following equations
\begin{equation}
\mathbf{D}_{\alpha }\mathbf{e}_{\beta }=\mathbf{\Gamma }_{\ \alpha \beta
}^{\gamma }\mathbf{e}_{\gamma },\mbox{\ or \ }\mathbf{\Gamma }_{\ \alpha
\beta }^{\gamma }\left( u\right) =\left( \mathbf{D}_{\alpha }\mathbf{e}%
_{\beta }\right) \rfloor \mathbf{e}^{\gamma },  \label{dcon1}
\end{equation}%
where, by definition, $L_{jk}^{i}=\left( \mathbf{D}_{k}\mathbf{e}_{j}\right)
\rfloor e^{i},\quad L_{bk}^{a}=\left( \mathbf{D}_{k}e_{b}\right) \rfloor
\mathbf{e}^{a},~C_{jc}^{i}=\left( \mathbf{D}_{c}\mathbf{e}_{j}\right)
\rfloor e^{i},$ $C_{bc}^{a}=\left( \mathbf{D}_{c}e_{b}\right) \rfloor
\mathbf{e}^{a} $ are computed for N--adapted frames (\ref{ddif}) and (\ref%
{dder}).

In the subclass of d--connections $\mathbf{D}$ on $\mathbf{V,}$ for standard
physical applications, it is convenient to work with d--metric compatible
d--connections (metrical d--connections) satisfying the condition $\mathbf{%
Dg=0}$ including all h- and v-projections $%
D_{j}g_{kl}=0,D_{a}g_{kl}=0,D_{j}h_{ab}=0,D_{a}h_{bc}=0.$

The torsion of a d--connection $\mathbf{D=}(hD,\ vD),$ for any d--vectors $%
\mathbf{X,Y}$ is defined by the d--tensor field
\begin{equation}
\mathbf{T(X,Y)\doteqdot \mathbf{D}_{\mathbf{X}}Y-D}_{\mathbf{Y}}\mathbf{%
X-[X,Y].}  \label{tors1}
\end{equation}

One has a N--adapted decomposition
\begin{equation*}
\mathbf{T(X,Y)=T(}hX,hY\mathbf{)+T(}hX,\ vY\mathbf{)+T(}vX,hY\mathbf{)+T(}%
vX,\ vY).
\end{equation*}%
The d--torsions $hT(hX,hY),vT(vX,vY),...$ are called respectively the $h$ $%
(hh)$--torsion, $v~(vv)$--torsion and so on.

We can also consider a N--adapted differential 1--form $\mathbf{\Gamma }_{\
\beta }^{\alpha }=\mathbf{\Gamma }_{\ \beta \gamma }^{\alpha }\mathbf{e}%
^{\gamma },$ %\label{dconf}
from which we can compute the torsion $\mathcal{T}^{\alpha }\doteqdot
\mathbf{De}^{\alpha }=d\mathbf{e}^{\alpha }+\Gamma _{\ \beta }^{\alpha
}\wedge \mathbf{e}^{\beta }.$ Locally, we get the (N--adapted) d--torsion
coefficients
\begin{eqnarray}
T_{\ jk}^{i} &=&L_{\ jk}^{i}-L_{\ kj}^{i},\ T_{\ ja}^{i}=-T_{\ aj}^{i}=C_{\
ja}^{i},\ T_{\ ji}^{a}=\Omega _{\ ji}^{a},\   \notag \\
T_{\ bi}^{a} &=&-T_{\ ib}^{a}=\frac{\partial N_{i}^{a}}{\partial y^{b}}-L_{\
bi}^{a},\ T_{\ bc}^{a}=C_{\ bc}^{a}-C_{\ cb}^{a}.  \label{dtors}
\end{eqnarray}

The curvature of a d--connection $\mathbf{D}$ is defined
\begin{equation*}
\mathbf{R(X,Y)\doteqdot \mathbf{D}_{\mathbf{X}}\mathbf{D}_{\mathbf{Y}}-D}_{%
\mathbf{Y}}\mathbf{D}_{\mathbf{X}}\mathbf{-D}_{\mathbf{[X,Y]}}
\end{equation*}%
for any d--vectors $\mathbf{X,Y.}$ By a straightforward d--form calculus, we
can find the N--adapted components of the curvature
\begin{equation}
\mathcal{R}_{~\beta }^{\alpha }\doteqdot \mathbf{D\Gamma }_{\ \beta
}^{\alpha }=d\mathbf{\Gamma }_{\ \beta }^{\alpha }-\mathbf{\Gamma }_{\ \beta
}^{\gamma }\wedge \mathbf{\Gamma }_{\ \gamma }^{\alpha }=\mathbf{R}_{\ \beta
\gamma \delta }^{\alpha }\mathbf{e}^{\gamma }\wedge \mathbf{e}^{\delta },
\label{curv}
\end{equation}%
of a d--connection $\mathbf{D},$ i.e. the d--curvatures:
\begin{eqnarray*}
R_{\ hjk}^{i} &=&e_{k}L_{\ hj}^{i}-e_{j}L_{\ hk}^{i}+L_{\ hj}^{m}L_{\
mk}^{i}-L_{\ hk}^{m}L_{\ mj}^{i}-C_{\ ha}^{i}\Omega _{\ kj}^{a},  \notag \\
R_{\ bjk}^{a} &=&e_{k}L_{\ bj}^{a}-e_{j}L_{\ bk}^{a}+L_{\ bj}^{c}L_{\
ck}^{a}-L_{\ bk}^{c}L_{\ cj}^{a}-C_{\ bc}^{a}\Omega _{\ kj}^{c},  \notag \\
R_{\ jka}^{i} &=&e_{a}L_{\ jk}^{i}-D_{k}C_{\ ja}^{i}+C_{\ jb}^{i}T_{\
ka}^{b},  \label{dcurv} \\
R_{\ bka}^{c} &=&e_{a}L_{\ bk}^{c}-D_{k}C_{\ ba}^{c}+C_{\ bd}^{c}T_{\
ka}^{c},  \notag \\
R_{\ jbc}^{i} &=&e_{c}C_{\ jb}^{i}-e_{b}C_{\ jc}^{i}+C_{\ jb}^{h}C_{\
hc}^{i}-C_{\ jc}^{h}C_{\ hb}^{i},  \notag \\
R_{\ bcd}^{a} &=&e_{d}C_{\ bc}^{a}-e_{c}C_{\ bd}^{a}+C_{\ bc}^{e}C_{\
ed}^{a}-C_{\ bd}^{e}C_{\ ec}^{a}.  \notag
\end{eqnarray*}

The Ricci tensor $\mathbf{R}_{\alpha \beta }\doteqdot \mathbf{R}_{\ \alpha
\beta \tau }^{\tau }$ is characterized by h- v--components, i.e. d--tensors,%
\begin{equation}
R_{ij}\doteqdot R_{\ ijk}^{k},\ \ R_{ia}\doteqdot -R_{\ ika}^{k},\
R_{ai}\doteqdot R_{\ aib}^{b},\ R_{ab}\doteqdot R_{\ abc}^{c}.
\label{dricci}
\end{equation}%
The scalar curvature of a d--connection is
\begin{equation}
\ _{s}\mathbf{R}\doteqdot \mathbf{g}^{\alpha \beta }\mathbf{R}_{\alpha \beta
}=g^{ij}R_{ij}+h^{ab}R_{ab},  \label{sdccurv}
\end{equation}%
defined by a sum the h-- and v--components of (\ref{dricci}) and d--metric (%
\ref{m1}).

For any metric structure $\mathbf{g}$ on a manifold $\mathbf{V,}$ there is
the unique metric compatible and torsionless Levi--Civita connection $\nabla $
for which $\ ^{\nabla }\mathcal{T}^{\alpha }=0$ and $\nabla \mathbf{g=0.}$
This is not a d--connection because it does not preserve under parallelism the N--connection splitting (\ref{whitney}) (it is not adapted to the N--connection structure).

\begin{theorem}
\label{thcdc}For any d--metric $\mathbf{g}=[hg,vg]$ on a N--anholonomic
manifold $\mathbf{V,}$ there is a unique metric canonical d--connection $%
\widehat{\mathbf{D}}$ satisfying the conditions $\widehat{\mathbf{D}}\mathbf{%
g=}0$ and with vanishing $h(hh)$--torsion, $v(vv)$--torsion, i. e. $h%
\widehat{T}(hX,hY)=0$ and $\mathbf{\ }v\widehat{T}(vX,\mathbf{\ }vY)=0.$
\end{theorem}

\begin{proof}
By straightforward calculations, we can verify that the d--connec\-ti\-on
with coefficients $\widehat{\mathbf{\Gamma }}_{\ \alpha \beta }^{\gamma
}=\left( \widehat{L}_{jk}^{i},\widehat{L}_{bk}^{a},\widehat{C}_{jc}^{i},%
\widehat{C}_{bc}^{a}\right) ,$ for
\begin{eqnarray}
\widehat{L}_{jk}^{i} &=&\frac{1}{2}g^{ir}\left(
e_{k}g_{jr}+e_{j}g_{kr}-e_{r}g_{jk}\right) ,  \label{candcon} \\
\widehat{L}_{bk}^{a} &=&e_{b}(N_{k}^{a})+\frac{1}{2}h^{ac}\left(
e_{k}h_{bc}-h_{dc}\ e_{b}N_{k}^{d}-h_{db}\ e_{c}N_{k}^{d}\right) ,  \notag \\
\widehat{C}_{jc}^{i} &=&\frac{1}{2}g^{ik}e_{c}g_{jk},\ \widehat{C}_{bc}^{a}=%
\frac{1}{2}h^{ad}\left( e_{c}h_{bd}+e_{c}h_{cd}-e_{d}h_{bc}\right) .  \notag
\end{eqnarray}
satisfies the condition of Theorem.$\square $
\end{proof}

In modern classical and quantum gravity theories defined by a (pseudo)
Riemannian metric structure $\mathbf{g}$ (\ref{metr}), it is preferred to
work only with the Levi--Civita connection $\nabla (\mathbf{g})=\{~\
_{\shortmid }\Gamma (\mathbf{g})\},$ which is uniquely defined by this
metric structure. Nevertheless, for a given N--connection splitting $\mathbf{%
N}$ on a nonholonomic manifold $\mathbf{V,}$ with redefinition of the metric
structure in the form $\mathbf{g}$ (\ref{ansatz}), there is an infinite
number of metric compatible d--connections uniquely defined by $\mathbf{g}$ (%
\ref{metr}), equivalently by $\mathbf{g}$ (\ref{ansatz}) and $\mathbf{N.}$

The Levi--Civita linear connection $\bigtriangledown =\{\Gamma _{\beta
\gamma }^{\alpha }\},$ uniquely defined by the conditions $~\mathcal{T}=0$
and $\bigtriangledown g=0,$ is not adapted to the distribution (\ref{whitney}%
). There is an extension of the Levi--Civita connection $\nabla $ to a
canonical d--connection $\widehat{\mathbf{D}}=\{\widehat{\mathbf{\Gamma }}%
_{\ \alpha \beta }^{\gamma }\}$ (\ref{candcon}), which is metric compatible
and defined only by a metric $\mathbf{g}$ when $\widehat{T}_{\ jk}^{i}=0$
and $\widehat{T}_{\ bc}^{a}=0$ but $\widehat{T}_{\ ja}^{i},\widehat{T}_{\
ji}^{a}$ and $\widehat{T}_{\ bi}^{a}$ are not zero, see (\ref{dtors}).

A straightforward calculus shows that the coefficients of the Levi--Civita
connection can be expressed in the form
\begin{equation}
\ \Gamma _{\ \alpha \beta }^{\gamma }=\widehat{\mathbf{\Gamma }}_{\ \alpha
\beta }^{\gamma }+Z_{\ \alpha \beta }^{\gamma },  \label{cdeft}
\end{equation}%
where
\begin{eqnarray}
Z_{jk}^{i} &=&0,\ Z_{jk}^{a}=-C_{jb}^{i}g_{ik}h^{ab}-\frac{1}{2}\Omega
_{jk}^{a},~Z_{bk}^{i}=\frac{1}{2}\Omega _{jk}^{c}h_{cb}g^{ji}-\Xi
_{jk}^{ih}~C_{hb}^{j},  \notag \\
Z_{bk}^{a} &=&~^{+}\Xi _{cd}^{ab}~\left[ L_{bk}^{c}-e_{b}(N_{k}^{c})\right]
,\ Z_{kb}^{i}=\frac{1}{2}\Omega _{jk}^{a}h_{cb}g^{ji}+\Xi
_{jk}^{ih}~C_{hb}^{j},  \label{cdeftc} \\
\ Z_{jb}^{a} &=&-~^{-}\Xi _{cb}^{ad}~~^{\circ }L_{dj}^{c},\
Z_{bc}^{a}=0,Z_{ab}^{i}=-\frac{g^{ij}}{2}\left[ ~^{\circ
}L_{aj}^{c}h_{cb}+~^{\circ }L_{bj}^{c}h_{ca}\right] ,  \notag \\
\Xi _{jk}^{ih} &=&\frac{1}{2}(\delta _{j}^{i}\delta
_{k}^{h}-g_{jk}g^{ih}),~^{\pm }\Xi _{cd}^{ab}=\frac{1}{2}(\delta
_{c}^{a}\delta _{d}^{b}+h_{cd}h^{ab}),  \notag
\end{eqnarray}%
for $\Omega _{jk}^{a}$ computed as in formula (\ref{ncurv}), $~^{\circ
}L_{aj}^{c}=L_{aj}^{c}-e_{a}(N_{j}^{c})$ and
\begin{eqnarray*}
\Gamma _{\beta \gamma }^{\alpha } &=&\left(
L_{jk}^{i},L_{jk}^{a},L_{bk}^{i},L_{bk}^{a},C_{jb}^{i},C_{jb}^{a},C_{bc}^{i},C_{bc}^{a}\right) ,
\\
\bigtriangledown _{\mathbf{e}_{k}}(\mathbf{e}_{j}) &=&~L_{jk}^{i}\mathbf{e}%
_{i}+L_{jk}^{a}e_{a},\bigtriangledown _{\mathbf{e}_{k}}(e_{b})=L_{bk}^{i}%
\mathbf{e}_{i}+L_{bk}^{a}e_{a}, \\
\bigtriangledown _{e_{b}}(\mathbf{e}_{j}) &=&~C_{jb}^{i}\mathbf{e}%
_{i}+C_{jb}^{a}e_{a},\bigtriangledown _{e_{c}}(e_{b})=C_{bc}^{i}\mathbf{e}%
_{i}+C_{bc}^{a}e_{a}.
\end{eqnarray*}%
It should be emphasized that all components of $\ \Gamma _{\ \alpha \beta
}^{\gamma },\widehat{\mathbf{\Gamma }}_{\ \alpha \beta }^{\gamma }$ and$\
Z_{\ \alpha \beta }^{\gamma }$ are defined by the coefficients of \
d--metric $\mathbf{g}$ (\ref{m1}) and N--connection $\mathbf{N}$ (\ref%
{coeffnc}), or equivalently by the coefficients of the corresponding generic
off--diagonal metric\ (\ref{ansatz}).

For instance, such a principle can be defined by any condition to construct
from the given metric coefficients and a $(n+n)$--splitting a unique
d--connection compatible to the canonical almost complex structure (this is
the so--called Cartan connection), or admitting a straightforward
application of Fedosov quantization in Einstein gravity, or of
Finsler--Lagrange geometry. Such constructions were recently developed in
order to construct more general classes of exact solutions in gravity \cite%
{vijmmp,vncg} (see also Part II in \cite{vsgg}), physical applications of
Ricci flow theory \cite{vrf,vv1,vnhrf2} and Fedosov quantization of Einstein
gravity in almost K\"{a}hler and/or Finsler--Lagrange variables \cite%
{vpla,vlgr}.

\end{document}